\documentclass[11pt]{article}
\usepackage{amsmath,amsthm,amssymb,amsfonts}
\usepackage[english]{babel}
\usepackage[utf8]{inputenc}
\usepackage[T1]{fontenc}
\usepackage{color}
\usepackage{hyperref}
\usepackage{graphicx}
\usepackage{comment}
\usepackage{wrapfig}
\usepackage{subcaption}
\usepackage{enumerate}
\usepackage{authblk}
\usepackage{geometry}
\usepackage[linesnumbered,ruled]{algorithm2e}
\usepackage{pdfpages}
\usepackage{tabularx}
\usepackage{epsfig}

\usepackage{authblk}
\usepackage{thm-restate}

\usepackage[font=small]{caption} 

\renewcommand{\ge}{\geqslant}
\renewcommand{\le}{\leqslant}
\newcommand{\vep}{\varepsilon}

\newcommand{\ol}{\overline}
\newcommand{\tf}{\textsf}

\def\ie{\emph{i.e. }}
\def\resp{\emph{resp. }}

\def\calB{\mathcal{B}}

\def\calE{\mathcal{E}}

\def\calO{\mathcal{O}}
\def\calG{\mathcal{G}}

\newcommand {\ignore} [1] {}

\newcommand\pr[1]{\textsc{#1}}
\newcommand\algo[1]{\textbf{#1}}

\def\Greedystar{\algo{Greedy$^{\bigstar}$} }
\def\Greedy{\algo{Greedy} }

\newtheoremstyle{mytheoremstyle} 
    {\topsep}                    
   {\topsep}                    
   {\itshape}                   
    {}                           
   {\scshape}                   
    {.}                          
    {.5em}                       
    {}  

\theoremstyle{mytheoremstyle}

\newtheorem{claim}{Claim}
\newtheorem{theo}[claim]{Theorem}
\newtheorem{lem}[claim]{Lemma}

\newtheorem{cor}[claim]{Corollary}

\newtheorem{df}[claim]{Definition}
\newtheorem{obs}[claim]{Observation}

\newcommand\piotr[1]{\textcolor{green}{\textrm{[PIOTR : #1]}}}
\newcommand\mathieu[1]{\textcolor{red}{\textrm{[MATHIEU : #1]}}}
\newcommand\nan[1]{\textcolor{blue}{\textrm{[NAN : #1]}}}
\newcommand\nann[1]{\textcolor{blue}{\textrm{[NAN-linear : #1]}}}

\title{Ultimate greedy approximation of independent sets in subcubic graphs}

\author[1]{Piotr Krysta}
\author[2]{Mathieu Mari}
\author[1]{Nan Zhi}

\affil[1]{University of Liverpool, Liverpool, UK. {\tt pkrysta@liverpool.ac.uk, n.zhi@liverpool.ac.uk}.}
\affil[2]{École Normale Supérieure, Université PSL, Paris, France. {\tt mathieu.mari@ens.fr}.}

\date{9 July 2019}

\begin{document}

\maketitle
\setcounter{page}{0}
\thispagestyle{empty}

\begin{abstract}
We study the approximability of the maximum size independent set (MIS) problem in bounded degree graphs.  This is one of the most classic and widely studied NP-hard optimization problems. It is known for its inherent hardness of approximation.

We focus on the well known minimum degree greedy algorithm for this problem. This algorithm iteratively chooses a minimum degree vertex in the graph, adds it to the solution and removes its neighbors, until the remaining graph is empty. The approximation ratios of this algorithm have been very widely studied, where it is augmented with an advice that tells the greedy which minimum degree vertex to choose if it is not unique.

Our main contribution is a new mathematical theory for the design of such greedy algorithms with efficiently computable advice and for the analysis of their approximation ratios. With this new theory we obtain the ultimate approximation ratio of 5/4 for greedy on graphs with maximum degree 3, which completely solves the open problem from the paper by Halld{\'o}rsson and Yoshihara (1995). Our algorithm is the fastest currently known algorithm with this approximation ratio on such graphs. We also obtain a simple and short proof of the (D+2)/3-approximation ratio of any greedy on graphs with maximum degree D, the result proved previously by Halld{\'o}rsson and Radhakrishnan (1994). We almost match this ratio by showing a lower bound of (D+1)/3 on the ratio of any greedy algorithm that can use any advice. We apply our new algorithm to the minimum vertex cover problem on graphs with maximum degree 3 to obtain a substantially faster 6/5-approximation algorithm than the one currently known.

We complement our positive, upper bound results with negative, lower bound results which prove that the problem of designing good advice for greedy is computationally hard 
and even hard to approximate on various classes of graphs. These results significantly improve on such previously known hardness results. Moreover, these results suggest that obtaining the upper bound results on the design and analysis of greedy advice is non-trivial.
\end{abstract}

\newpage




 

\section{Introduction}

Given an undirected graph $G$, an {\em independent set} in $G$ is a subset of the set of its vertices such that no two of these vertices are connected by an edge in $G$.
The problem of finding an independent set of maximum cardinality in a graph, the Maximum Independent Set problem (MIS), is one of the most fundamental NP-hard
combinatorial optimization problems. Already Karp proved in his famous paper \cite{Karp72} that the decision version of the Maximum Clique problem, which is equivalent to MIS on the complement graph, is NP-complete. Because of its hardness of computation, we are interested in polynomial time approximation algorithms for the MIS problem. We say that a polynomial time algorithm for MIS is an $r$-approximation algorithm if it finds an independent set in the input graph of size at least $opt/r$, where $opt$ is the size of the maximum size independent set in the graph. The number $r \geq 1$, which may be constant or may depend on the input graph's parameters, is called an {\em approximation ratio}, {\em guarantee} or {\em factor}.

We are interested in MIS on graphs with maximum degree bounded by $\Delta$. This problem is known for its inherent hardness of approximation guarantee. Even if $\Delta =3$, MIS is known to be APX-hard, see \cite{AlimontiK00}. There are also explicit constant hardness ratios known for small constant values of $\Delta$ \cite{Chlebiks2003}. As $\Delta$ grows, there are stronger, asymptotic, hardness of approximation results known: $\Omega(\Delta/ \log^2 \Delta)$, under the Unique Games Conjecture \cite{AKS2011}, and $\Omega(\Delta/ \log^4 \Delta)$,  assuming that P $\not =$ NP \cite{SOChan16}. The best known polynomial time approximation ratio for this problem for small values of $\Delta \geq 3$ is arbitrarily close to $\frac{\Delta +3}{5}$, see \cite{Berman1999,Berman:1994:AMI:314464.314570,Chlebik2004}. This is achieved by a local search approach at the expense of huge running time, e.g., $n^{50}$ \cite{HalldorssonRadha1994}, where $n$ is the number of vertices in the graph. The best known asymptotic polynomial time approximation ratio for MIS is $O(\Delta \log \log (\Delta)/\log(\Delta))$ based on semidefinite programming relaxation \cite{Halperin02}. However, the best known asymptotic  approximation ratio for MIS is $O(\Delta / \log^2(\Delta))$ with $O(n^{O(1)} \cdot 2^{O(\Delta)})$ running time \cite{Bansal:2015:LTF:2746539.2746607}. In this paper we are primarily interested in MIS on graphs with small to moderate values of $\Delta$.

Probably the best known algorithmic paradigm to find large independent sets is the {\em minimum degree greedy} method, which repeatedly chooses a minimum degree vertex in the current graph as part of the solution and deletes it and its neighbors until the remaining graph is empty. This basic algorithm is profoundly simple and time-efficient and can be implemented to run in linear time. The first published approximation guarantee $\Delta+1$ of this greedy algorithm for MIS we are aware of can be inferred from the proof of the following conjecture of Erd\H{o}s, due to Hajnal and Szemer{\'e}di \cite{HajnalS1970,Berge1973}: every graph with $n$ vertices and maximum degree $\Delta$ can be partitioned into $\Delta + 1$ disjoint independent sets of almost equal sizes. The approximation ratio of greedy has been improved to $\Delta - 1$ by Simon \cite{Simon90}. The best known analysis of greedy by Halld{\'o}rsson and Radhakrishnan \cite{Halldorsson1997,HalldorssonR1994} for MIS implies the approximation ratio of $(\Delta+2)/3$, and better ratios are known for small values of $\Delta$.

Halld{\'o}rsson and Yoshihara \cite{10.1007/BFb0015418} asked in their paper the following fundamental question: what is the power of the greedy algorithm when we augment it with an {\em advice}, that is, a fast method that tells the greedy which minimum degree vertex to choose if there are many? They, for instance, proved that no advice can imply a better than $5/4$-approximation of greedy for MIS with $\Delta=3$. On the other hand, they provide an advice for greedy that implies a $3/2$-approximation, see algorithm MoreEdges in \cite{10.1007/BFb0015418}, and an improved $9/7$-approximation, see algorithm Simplicial in \cite{10.1007/BFb0015418}.\footnote{We have found a counter-example to their claimed ratio of $9/7$ by Simplicial, see example in Figure \ref{fig:lowerbound}. Simplicial may choose recursively the top vertex in those instances, which leads to a solution where the approximation ratio tends to $17/13 > 9/7$, when $i$ tends to infinity. This counter-example has also been verified and confirmed by Halld{\'o}rsson \cite{Halldorsson2019}. Using our new techniques, we can prove that Simplicial achieves an approximation ratio of $\frac{13}{9} \approx 1.444$ for MIS on subcubic graphs. Technically, Simplicial is not a greedy algorithm however, because due to its branchy reduction it might have iterations where it does not choose the current minimum degree vertex \cite{Halldorsson2019}.} In fact this results has been retracted by Halld{\'o}rsson, see \cite{Halldorsson2019}. Thus, $3/2$ is the best known to date bound on the approximation ratio of greedy in graphs with maximum degree at most $3$, which are also called {\em subcubic} graphs.

\smallskip

\noindent
{\bf Motivation.} In addition to its simplicity and time efficiency, the greedy algorithm for MIS is also important in its own right. Following Halld{\'o}rsson and Radhakrishnan \cite{Halldorsson1997}, greedy algorithm is known to have several important properties: it achieves the celebrated Tur{\'a}n bound \cite{Turan1941,Erdos1970}, and its generalization in terms of degree sequences \cite{Wei1981}, it achieves a good graph coloring approximation when applied iteratively as a coloring method \cite{DSJohnson1974}. Finally, the greedy algorithm finds optimal independent sets in trees, complete graphs, series-parallel graphs, co-graphs, split graphs, $k$-regular bipartite graphs, and graphs with maximum degree at most $2$ \cite{Halldorsson1997,BODLAENDER1997101}. Another important but non-explicit class of graphs for which greedy is optimal is the class of
{\em well-covered} graphs, introduced by Plummer \cite{Plummer70}, and widely studied, see \cite{Plummer93} for a survey. A graph
is well-covered if all its maximal independent sets have the same size. In particular, because any greedy
set is maximal, the greedy algorithm is optimal on such graphs.
 Furthermore, the greedy algorithm finds frequent applications in graph theory, helping to prove that certain classes of graphs have large independent sets, e.g., it almost always finds a $2$-approximation to MIS in a random graph \cite{McDiarmid84}, or it provides an independent set of size at least $0.432 n$ in random cubic graphs with probability tending to $1$ as $n$, the number of vertices, tends to $+\infty$ \cite{FriezeS94}.

\subsection{Our new results}

\noindent{\bf Positive results: upper bounds.} We study the design and analysis of greedy approximation algorithms with advice for MIS on bounded degree graphs. Our main technical contribution is a new class of payment schemes for proving improved and tight approximation ratios of greedy with advice. With our new payment schemes we obtain the best known analyses of the greedy algorithm on bounded degree graphs, which significantly improve on the previously known analyses. As a warm-up, we first apply these new techniques to MIS on graphs with maximum degree bounded by any $\Delta$ to obtain the following results:
\begin{itemize}
\item A simple and short proof of the $(\Delta+2)/3$-approximation ratio of any greedy algorithm (i.e., without any advice), the result proved previously by Halld{\'o}rsson and Radhakrishnan \cite{Halldorsson1997,HalldorssonR1994}. We extend a lower bound construction of Halld{\'o}rsson and Radhakrishnan \cite{Halldorsson1997} to prove that any greedy algorithm (with any, even exponential time, advice) has an approximation ratio at least $(\Delta+1)/3 - O(1/\Delta)$.
\item A simple proof of the $(\Delta+6)/4$-approximation ratio of any greedy algorithm on triangle-free graphs with maximum degree $\Delta$, which improves the previous best known greedy ratio of $\Delta/3.5 + O(1)$ \cite{HalldorssonR1994} for MIS on triangle-free graphs. Compared to the proof in \cite{HalldorssonR1994} which uses a technique of Shearer \cite{Shearer83}, our proof is extremely simple and short.
\end{itemize}

We see that as $\Delta$ increases, there is no hope in obtaining significantly better approximation than $(\Delta+2)/3$ by using any, even exponential time, advice for greedy. This motivates us to focus on the small values of $\Delta$. Indeed, we have to develop our payment scheme techniques significantly more compared to the above applications to MIS on graphs with maximum degree $\Delta$ for any value of $\Delta$.
In particular, we obtain the following results for MIS and Minimum Vertex Cover (MVC) problems:
\begin{itemize}
\item We completely resolve the open problem from the paper of Halld{\'o}rsson and Yoshihara \cite{10.1007/BFb0015418} and design a fast, ultimate advice for greedy obtaining a $5/4$-approximation, that is, the best possible greedy ratio for MIS on subcubic graphs. A lower bound of $5/4$ on the ratio of greedy with any, even exponential time, advice on such graphs was proved in \cite{10.1007/BFb0015418}, and the best previously known ratio of greedy was $3/2$ \cite{10.1007/BFb0015418}. Halld{\'o}rsson and Radhakrishnan \cite{Halldorsson1997} also prove a lower bound of $5/3$ for any greedy algorithm that does not use any advice for MIS on subcubic graphs. 
Our new greedy
$5/4$-approximation algorithm has running time $O(n^2)$, where $n$ is the number of vertices in the graph. For comparison, the best known algorithm for this problem is a local search $6/5$-approximation algorithm of Berman and Fujito \cite{Berman1999}, and with an analysis from \cite{HalldorssonRadha1994} has a running time no less than $n^{50}$. Specifically, if the approximation ratio of this local search algorithm is fixed to $5/4$, then the running time is $n^{18.27}$, see \cite{Chlebik2004}.
\item We obtain a greedy $4/3$-approximation algorithm for MIS on subcubic graphs with linear running time, $O(n)$. By using our payment scheme, we can also provide a simple proof of a $3/2$-approximation ratio of the greedy algorithm called MoreEdges in \cite{HalldorssonRadha1994}, which was the best previously known approximation ratio of greedy for MIS on subcubic graphs.
\item Then, we also
obtain a fast  $O(n^2)$-time $6/5$-approximation for the MVC problem on subcubic graphs. The previous best algorithm for this problem was a $7/6$-approximation with a running time of at least $n^{50}$ \cite{HalldorssonRadha1994}. Even obtaining the $6/5$-approximation for MVC on subcubic graphs required a running time of $n^{18.27}$ \cite{Chlebik2004}.
\end{itemize}

To prove these results we develop a payment technique to pay for the greedy solution via a specially defined class of potential functions. For this new class of potentials on subcubic graphs, we develop a very specific inductive process, which takes into account ``parities'' and priorities of the reductions performed by greedy, to prove that the value of the potential is kept locally to be at least $-1$. An additional, global argument is required to show that the global potential is at least $0$. For more details about our new techniques, see Section \ref{s:tech-contributions}. \\

\noindent{\bf Negative results: lower bounds.} We complement our positive upper bound results with impossibility, lower bounds, results which suggest that our upper bounds on the design of good advice for greedy are essentially (close to) best possible, or non-trivial computational problems. We believe that this also suggests that the design of good advice for greedy is a non-trivial task on its own.

Let us first observe that a solution output by greedy is a  maximal independent set. A graph is called {\em well-covered} if all of its maximal independent sets are of the same size, see \cite{Plummer70,Plummer93}. Caro et al.~\cite{CARO1996137} study the computational  complexity of the problem of deciding if a given graph is well-covered. They prove that this problem is co-NP-complete even on $K_{1,4}$-free graphs.

To prove our lower bounds we resort to a notion which captures the essence of greedy (the well-covered property reveals only a very restricted feature of greedy). Namely, we study the computational complexity of computing a good advice for the greedy algorithm for MIS. Towards this goal, Bodlaender et al.~\cite{BODLAENDER1997101} defined a problem called MaxGreedy, which given an input graph asks for finding the largest possible independent set obtained by any greedy algorithm. Thus, MaxGreedy asks for computing the best advice for greedy, i.e., one that leads to the largest possible greedy independent set.
They proved that the problem of computing an advice which finds an $r$-approximate solution to the MaxGreedy problem is co-NP-hard for any fixed rational number $r \geq 1$ and that this problem with $r=1$ remains NP-complete \cite{BODLAENDER1997101}. 

We significantly improve the previously known results on the hardness of computing good advice for greedy, by obtaining the following new results:
\begin{itemize}
    \item We prove that the MaxGreedy problem is NP-complete even on cubic planar graphs. This significantly strengthens the NP-completeness result by Bodlaender et al.~\cite{BODLAENDER1997101} who prove it on arbitrary, not even bounded-degree, graphs. This result suggests that the problem of designing and analysing good advice for greedy even on cubic planar graphs is difficult.
    \item We further prove that MaxGreedy is even NP-hard to approximate to within a ratio of $n^{1-\varepsilon}$ for any $\varepsilon > 0$ by a reduction from 3-SAT, and hard to approximate to within $n/\log n$ under the exponential time hypothesis. We extend this construction to the class of graphs with bounded degree $\Delta$. We prove that MaxGreedy remains hard to approximate to within a factor $(\Delta+1)/3-O(1/\Delta)-O(1/n)$ on this class, nearly matching the approximation ratio $(\Delta+2)/3$ of the greedy algorithm in this class.
    \item We prove that the MaxGreedy problem remains hard to approximate on bipartite graphs. This is quite interesting because it is well known that the MIS problem is polynomially solvable on bipartite graphs.
\end{itemize}

Finally, we extend a lower bound construction of Halld{\'o}rsson and Radhakrishnan \cite{Halldorsson1997} to prove that any greedy algorithm (with any, even exponential time, advice) has an approximation ratio at least $(\Delta+1)/3 - O(1/\Delta)$ on graphs with maximum degree $\Delta$.

\subsection{Our technical contributions}\label{s:tech-contributions}

Our main technical contributions are a class of potential functions and payment schemes, together with an inductive proof technique that are used to pay for solutions of greedy algorithms for MIS. These new techniques lead to very precise and tight analyses of the approximation ratios of greedy algorithms.  

Here we will explain intuitions about our proof of the $5/4$-approximation ratio of the greedy algorithms on subcubic graphs, which uses the full technical machinery of our approach.
%
%
Let $G$ be a given input graph with an optimal independent set $OPT$. Greedy algorithm executes {\em reductions} on $G$, i.e., a reduction is to pick a minimum degree vertex $v$ in the current graph ({\em root} of the reduction) into the solution and remove its neighbors, see Figure \ref{fig:basicreductions} for examples of reductions. Suppose the first reduction executed by greedy is $R$ and it is {\em bad}: its root $v$ has degree $2$, $v \not \in OPT$ and both neighbors of $v$ are in $OPT$. Then, locally, the approximation ratio is $2$. To bring the approximation ratio down to $5/4$, we must prove that, in the future, there will exist equivalent of at least three reductions, called {\em good}, each of which adds one vertex to the solution and removes only one vertex from $OPT$. Moreover, for each executed bad reduction, there must exists an unique (equivalent of) three good reductions.

Consider, for example, the family of instances of MIS in Figure \ref{fig:lowerbound} where the base graph $H_0$ is $H_0'$. There, black vertices belong to $OPT$, while white do not. This class of instances is due to Halld{\'o}rsson and Yoshihara \cite{10.1007/BFb0015418}. Any greedy algorithm executes on this instance many bad reductions, but only at the very end, good reductions, triangles, are executed. It can easily be checked that there is just enough good reductions to uniquely map three of those to any executed bad reduction (in fact in the whole process there is exactly one good reduction that is unused). This essentially shows a lower bound of $5/4$ on the ratio of any greedy when $i$ tends to infinity. We see that the ``payment'' for bad reductions arrives, but very late! Such a ``payment'' may not only be late, but we also do not know when ``good'' reductions providing such payment are executed. Thus good reductions might be very irregularly distributed. For instance, suppose that the first reduction in $H_0'$ on Figure \ref{fig:lowerbound}, let us call it $R$, has two of its contact edges (these are the four edges going down from $R$'s two black vertices) going to an identical white vertex, creating a follow up reduction of degree one. Then that degree one reduction is good and when executed, it can immediately (partially) pay for the bad reduction $R$.     

\smallskip

\noindent
{\bf Question:} How do we prove an existence of such a highly non-local and irregular payment scheme?
 We will
define a special potential of a reduction, see Section \ref{section:potential}, which will imply the existence of 
two sources of ``payments'' -- in the past, from the very first executed reduction, and -- in the future, from the executed good reductions.
  Moreover, our 
potential will be defined in such a way that each executed reduction can in some sense be ``almost paid for'' locally, so that at every point in time we will keep the value of the potential of each connected component of at least $-1$. For example in the instance from Figure \ref{fig:lowerbound} the execution of the first bad reduction in graph $H_{i+1}$ creates $4$ connected components each isomorphic to $H_i$ and then greedy executes reductions in each of them independently.

We will have an intricate inductive argument, see the Inductive Low-debt Lemma \ref{lemma:induction}, showing that an execution of a sequence of reductions in a connected input graph will have the total potential at least $-1$. In the induction step, some reductions $R$ may create multiple components each with potential $-1$. In such cases when we cannot locally obtain potential at least $-1$, we will make sure that even before the execution of $R$, such components contain reductions with strictly higher greedy priority than that of $R$, thus leading to a contradiction. Or, reduction $R$ can pay for such components.
  Having proved 
that the total potential of an execution in the (connected) input graph is at least $-1$, we will finally pay for this $-1$ by the very first executed reduction for which we will show that it will always possess an extra saving of $1$ (this is a payment from the past).

Intuitively, our potential will imply that we can ship the payments from good reductions executed {\em anywhere} in the graph by the greedy into the places where bad reductions need those payments. Such a shipment is unique, in the above sense that there exists three (equivalent) good reductions per single bad reduction.
 We will realize
this shipment by deferring the need of payment into the future along edges, called {\em contact edges}, which are incident to the neighbors of the reduction's root vertex. These contact edges created by a bad reduction $R$ will be called {\em loan} edges. Each loan edge $e$ created by $R$ will have a ``dual'' edge (physically identical to $e$), called a {\em debt} edge, which will be inherited by the future reductions directly created by $R$ via its contact edges.

Our potential of a reduction $R$ will account for the number of vertices chosen to the solution, removed from $OPT$, plus the number of loan edges, minus the number of debt edges. Thus, we will very precisely account for such edges.
  This process 
is complicated by the fact that vertices can be {\em black} (in $OPT$, a maximum independent set) and {\em white} (outside of $OPT$) and whether a reduction is ``bad'' depends on the distribution of black/white vertices in the reduction. For instance, a reduction like the first one in graph $H_0'$ on Figure \ref{fig:lowerbound} might have a black root and thus the two root's neighbors will be white. Such a reduction will in fact be ``good'' when executed.

Some reductions are ``bad'' and they create ``many'' contact edges, like the first reduction in graph $H_0'$ on Figure \ref{fig:lowerbound} (with white root). Those ``many'' contact edges, called loan edges, will in the future create some good reductions that will pay for that bad one. Observe that such a reduction has more loan edges than debt edges, so it creates a surplus of credit. On the other hand when the greedy process ends, it can end only with terminal reductions. Those terminal reductions do not have any contact edges, but they have the property that for any white vertex added to the solution, they remove only one black vertex. That is why they can ``pay'' for the previous bad reductions.

This explains only some intuitions of why such a highly non-local and irregular payment scheme can have a chance to succeed.
%
%
By this intuition our guess was (indeed, confirmed true by our final proof) that a single such contact edge, loan edge, will translate in a one-to-one way to a single (equivalent of a) good reduction.
Furthermore, and most importantly, this approach will enable us to ``predict'' the precise future graph structure by using the contact edges. It also enables us to keep track of the past reductions -- by keeping track of the debt edges and the current state of the savings. And indeed, we have succeeded in building a theory that delivers a complete and precise such payment scheme.

This approach allows us to achieve a very interesting kind of result here -- namely, to (essentially) characterize all graphs that can have negative potential! See the Definition \ref{df:problematic} and Lemma \ref{lemma:induction}.

These ideas lead to our analysis which is extremely tight, essentially up to an additive unit in the following sense. We prove that (a version of) our $5/4$-approximate greedy algorithm finds a solution of size at least $\frac{4}{5} |OPT| + \frac{1}{5}$ on any subcubic graph, whereas when run on the lower bound instances of Halld{\'o}rsson and Yoshihara \cite{10.1007/BFb0015418}, our algorithm finds a solution of size precisely $\frac{4}{5} |OPT| + \frac{1}{5}$.
  A somehow unusual 
aspect of our result is that we can prove that {\em any} lower bound example that shows exact tightness of our guarantee of  $\frac{4}{5} |OPT| + \frac{1}{5}$ must be an infinite family of graphs, see the remark after the proof of Theorem \ref{theo:ratio}.

\medskip

\noindent
{\bf Motivation.} To motivate our 
upper bound results further, we study the computational complexity of designing good advices for greedy. Given a class of graphs $\mathcal{G}$ there always exists a family of graphs $\mathcal{G}' \subset \mathcal{G}$ that implies the worst possible approximation performance of greedy with any, even exponential time advice.

\noindent 
{\bf Question:} To which extent we can reach this worst case performance by designing a polynomial time advice for greedy? As an example, we have answered this question completely and very precisely in the case where 
$\mathcal{G}$ is the class of all subcubic graphs and $\mathcal{G}'$ is the family of instances with base graph $H_0'$ and recursively constructed graphs $H_{i+1}$ in Figure \ref{fig:lowerbound}, due to Halld{\'o}rsson and Yoshihara \cite{10.1007/BFb0015418}. Interestingly, our algorithm, \Greedystar, outputs the best greedy independent set on those instances and we prove that this independent set is precisely a factor of $4/5$ times the size of the maximum independent set in the instance.
  However, we 
prove that in many cases this problem is highly computationally intractable, NP-hard, NP-complete, APX-hard and  
even hard to approximate with respect to possible approximation ratios achieved by greedy with advice, depending on the class of graphs $\mathcal{G}$, see Section \ref{s:lower_bounds} for details. For instance, the problem of designing the best advice for greedy is NP-complete even when $\mathcal{G}$ are all planar cubic graphs, see Theorem \ref{hardcubic}.

These computational complexity results suggest that the task of the design and analysis of efficient advice for the greedy algorithm is a non-trivial task. On the other hand, indeed, our ultimate analysis of the greedy on subcubic graphs, in Section \ref{Section:greedy-subcubic}, is quite complex.\\

\noindent
{\bf Differences between our proof and Halld{\'o}rsson and Yoshihara \cite{10.1007/BFb0015418}.} The success of our approach depends on three parts: the definition of the potential, the design of greedy rules and the analysis for excluding certain reductions and paying for some of them. These three parts interact in a very intricate way. Our definition of potential not only captures itself the graph structure of problematic reductions. Given a current reduction, our potential also captures, by the interaction with other two parts, what is the relation to reductions which were executed previously and will be executed in the future. This is captured by considering the different types of edges, called contact edges, which connect those reductions or sets of reductions.

According to our potential, there are two kinds of reductions that are particularly problematic to deal with. These are odd isolated cycles with maximum independent sets in them, and reductions like reduction (d) in Figure \ref{fig:negativepot}, which we call an odd-backbone reduction. Their potential is $-1$ (for the cycle it can also be $-2$, but we can prevent that case). This means that each such reduction when executed would need a unit of payment originated from some good reduction. Consider an instance $H_{i}$ in Figure \ref{fig:lowerbound} when $i$ tends to $\infty$ (with the base graph $H_0$). Suppose that greedy executes the top bad reduction and then recursively executes the following four created bad reductions. Then at the very end it will reach a collection of $7$-cycles and each such cycle will need a payment of $1$. As we see this can only lead to a ratio $17/13 > 5/4$. Already on any odd cycle, the potential of 
\cite{10.1007/BFb0015418} tells us that it actually needs a payment of $2$ (which is one unit more than our potential); we mention here however that it is not possible to pay $2$ units to such odd cycles. Our approach is to either prevent greedy from ever ending up with such isolated problematic odd cycles or to show that we can actually pay for such cycles in some cases. The key to a solution is to carefully prioritize certain reductions that would ``break'' the cycle before it becomes isolated, or to pay for it when there is a spare reduction that can do so. For the bad odd-backbone reductions in Figure \ref{fig:negativepot}(d), observe that we could wisely execute them on a black degree-$2$ vertex which would make then good. But then how do we know which of these two adjacent degree-$2$ vertices is black/white? In fact a ``branchy'' reduction of the Simplicial algorithm in \cite{10.1007/BFb0015418} deals with such reductions, but this makes the algorithm non-greedy.
We do not know how to use branchy to obtain a $5/4$-approximation, because it will introduce new degree-$3$ reductions which then would need to be analyzed.

One way, pursued in \cite{10.1007/BFb0015418}, is to try to pay for such odd cycles or odd-backbone reductions by some kind of local analysis which tries to collect locally good reductions that can pay. We can show that such local analysis/payment is not possible and a global payment or explicit exclusion of such reductions are necessary.

Instead, what we do is to impose a special greedy order on such odd-backbone reductions, and with this order we prove that we can pay for them whenever they are executed as bad reductions. The source of these payments, however, is non-local and our scheme proves their unique existence.

Now, to achieve the above payments or avoid bad reductions, we introduce a powerful analysis tool which is a special kind of reductions. They are called {\em black} and {\em white} reductions, see Definition \ref{df:type}. We also introduce an inductive process to argue about existence of such reductions in Lemma \ref{lemma:induction}. These techniques will let us prove that when a reduction, say $R$, that cannot pay is executed, there will exist a strictly higher priority reduction (black or white) in the graph even before the execution of $R$. This leads to a contradiction with the greedy order. This argument is quite delicate because their existence depends crucially on what kind of contact edges $R$ has. But it also depends on the previously executed reductions.

Most importantly, our definition of the potential is in perfect harmony with our inductive proof, that the potential can be kept locally at value at least $-1$. This lets us link the potential directly to the graph structure of the reductions, see Definition \ref{df:problematic}. And this then lets us characterize the potentially problematic graphs, that is, those with negative potential, which is the core of our proof. The main tool that helps us in this task is our Inductive Low-debt Lemma, see Lemma \ref{lemma:induction}, which enables us to design the greedy order and characterizes the problematic graphs.

Note that we have managed to prove the existence of appropriate payments coming from good reductions by using only ``local'' inductive arguments, but our payment scheme is inherently non-local. This means that the payment, i.e., good reductions, can reside very far from the bad reductions for which they pay.

\subsection{Further related work}


In this section we survey some further related work on MIS. It will be selective, because of vast existing literature on approximating MIS.
  The MIS problem 
is known for its notorious approximation hardness. H\r{a}stad \cite{hastad1999} provided a strong lower bound of $n^{1-\epsilon}$ on the approximation ratio of the general MIS problem for any $\epsilon>0$, under the assumption that NP $\not\subseteq$ ZPP, where $n$ is the number of vertices of the input graph. This hardness result has been derandomized by Zuckerman \cite{Zuckerman:2006:LDE:1132516.1132612} who showed that the general MIS is not approximable to within a factor of $n^{1-\epsilon}$ for any $\epsilon>0$, assuming that P $\not=$ NP.
The best known approximation algorithms for general MIS problem achieve the following approximation ratios: $O(n/\log^2n)$ by Boppana and Halld{\'o}rsson \cite{Boppana1992}, that was improved to $\widetilde{O}(n/\log^3n)$ by Feige \cite{Feige04} (this ratio has some additional $\log \log n$ factors).
 
The first known nontrivial approximation ratio for MIS on graphs with maximum degree $\Delta$ is $\Delta$ acquired by Lovasz's algorithmic proof \cite{LOVASZ1975269} of Brooks's coloring theorem.

The best known asymptotic polynomial time approximation ratio for MIS, i.e., when $\Delta$ is large, is $O(\Delta \log \log (\Delta)/\log(\Delta))$ based on semidefinite programming relaxation \cite{Halperin02}. However, if we allow for some extra time, there exists an asymptotic $O(\Delta/\log(\Delta))$-approximation
algorithm with $O(n^{O(1)} \cdot \exp(\log^{O(1)}n))$ running time \cite{Bansal15}. And in particular, the best known asymptotic approximation ratio for MIS is $O(\Delta / \log^2(\Delta))$ with $O(n^{O(1)} \cdot 2^{O(\Delta)})$ running time \cite{Bansal:2015:LTF:2746539.2746607}. For small to moderate values of $\Delta$, Halld{\'o}rsson and Radhakrishnan \cite{Halldorsson1994Removal,HalldorssonRadha1994}, via subgraph removal techniques, obtain an asymptotic ratio $\Delta/6(1 + o(1))$ with $O(\Delta^{O(1)} \cdot n)$ running time for relatively small $\Delta \geq 5$, and $O(\frac{\Delta}{\log\log \Delta})$ for larger $\Delta$ with linear running time.

Demange and Paschos \cite{DEMANGE1997105}  prove that a $\Delta/6$-approximation ratio can be obtained in time $O(n|E|)$, but this ratio is asymptotic as $\Delta \rightarrow \infty$. They also prove that an $\Delta/k$-approximation ratio can be achieved in time $O(n^{\lceil k/2\rceil})$, for any fixed integer $k$. This second ratio is also asymptotic. Those algorithms do not apply to $\Delta=3$, but to quite large values of $\Delta$. Khanna \emph{et al.} \cite{Khanna1998} obtained a $(\sqrt{8\Delta^2 + 4\Delta + 1} - 2\Delta + 1)/2$-approximation for $\Delta \geq 10$, with $O(\Delta^{O(1)} \cdot n)$ running time by local search.

For any values of  $\Delta$, Hochbaum \cite{HOCHBAUM1983243}, using a coloring technique accompanied with a method of Nemhauser and Trotter \cite{Nemhauser1975} obtained an algorithm with a ratio of $\Delta/2$. Berman and F\"{u}rer \cite{Berman:1994:AMI:314464.314570} designed a new algorithm whose performance ratios are arbitrarily close to $(\Delta + 3)/5$ for even $\Delta$ and $(\Delta + 3.25)/5$ for odd $\Delta$. Berman and Fujito \cite{Berman1999} obtained a better ratio which is arbitrarily close to $\frac{\Delta+3}{5}$. Finally, in the latest results from Chleb\'{i}k and Chleb\'{i}kov\'{a} \cite{Chlebik2004}, their approximation ratio is arbitrarily close to $\frac{\Delta+3}{5}- \frac{4(5\sqrt{13} - 18)}{5}\frac{(\Delta-2)!!}{(\Delta+1)!!}$, which is slightly better than the previous results. These algorithms are based on local search and they have huge running times.

For the case of subcubic graphs with $\Delta = 3$, MIS is known to be NP-hard to approximate  to within $\frac{95}{94}$. \cite{Chlebiks2003}. Hochbaum \cite{HOCHBAUM1983243} presented an algorithm with $3/2$ ratio. Berman and Fujito \cite{Berman1999} obtain a $\frac{6}{5}$ ratio with a huge running time. Even a tighter analysis from \cite{HalldorssonRadha1994} the time complexity appears to be no less than $n^{50}$. Chleb\'{i}k and Chleb\'{i}kova \cite{Chlebik2004} show that their approximation ratio is arbitrarily close to $3 - \frac{\sqrt{13}}{2}$, which is slightly better than $\frac{6}{5}$. Moreover, the time complexity of their algorithm is also better. Specially, if the ratio is fixed to $\frac{5}{4}$, then the running time is $n^{18.27}$. Halldorsson and Radhakrishnan \cite{Halldorsson1994Removal} provide another local search approach based on \cite{Berman:1994:AMI:314464.314570} and obtain a ratio of $\frac{7}{5}$ in linear time, and a $(\frac{4}{3} + \epsilon)$-ratio in time $O(n e^{1/\epsilon})$. Halld{\'o}rsson and Yoshihara \cite{10.1007/BFb0015418} present a linear time greedy algorithm with an  approximation ratio of $\frac{3}{2}$.\footnote{They also claim a better ratio of $9/7$ in linear time, however, they have retracted this result \cite{Halldorsson2019}.}

For the minimum vertex cover (MVC) problem in general, Garey and Johnson \cite{GAREY1976237} presented a $2$-approximation algorithm on general graphs. For MVC on subcubic graphs, Hochbaum \cite{HOCHBAUM1983243} provided a $\frac{4}{3}$-approximation ratio, by using the method of Nemhauser and Trotter \cite{Nemhauser1975}. Berman \cite{Berman:1994:AMI:314464.314570} gives a $\frac{7}{6}$ ratio and by the same approach. And \cite{Chlebik2004} shows that a ratio which is slightly better than $\frac{7}{6}$ can be obtained. These algorithms are based on local search and have huge running time.

\section{Definitions and preliminaries}

Given a graph $G=(V,E)$, we also denote $V(G) = V$ and $E(G)= E$. For a vertex $v\in V$, let $N_G(v):= \{u\in V \mid uv \in E\}$ and $N_G[v]:= N_G(v)\cup\{v\}$ denote respectively the open and closed neighborhood of $v$ in $G$. The degree of $v$ in $G$ denoted $d_G(v)$ is the size of its open neighborhood. 
More generally, we define the closed (\resp open) neighborhood of a subset $S\subseteq V$ as the union of all closed (\resp open) neighborhoods of each vertex in $S$. 

A graph is called {\em subcubic} or {\em sub-cubic} is its maximum degree is at most $3$. If the degree of each vertex in a graph is exactly $3$ then it is called {\em cubic}.

Given an independent set $I$ in $G$, we call \emph{black} vertex a vertex $v$ in $I$ and a \emph{white} vertex otherwise.  We denote by $\alpha(G)$ the independence number of $G$, that is the number of black vertices when $I$ is of maximum size in $G$.


\section{Greedy}

The greedy algorithm, called a {\em basic greedy}, or just \Greedy\hspace*{-1.7mm}, on a graph $G=(V,E)$ proceeds as follows. It starts with an empty set $S$. While the graph $G$ is non empty, it finds a vertex $v$ with minimum degree in the remaining graph, adds this vertex to $S$ and removes $v$ and its neighbors from $G$. It is clear that at the end, $S$ is an independent set. Let $S=\{v_1,\dots,v_k\}$ be the ordered output. Let $G_i$ denote the graph after removing vertex $v_i$ and its neighboring vertices. More precisely, $G_0=G$ and $G_i=G[V\setminus N_G\left[\{v_1,\dots,v_i\}]\right]$, where $v_i$ is a vertex in $G_{i-1}$ that satisfies $d_{G_{i-1}}(v_i)=\min\left\{d_{G_{i-1}}(v) : v \in V(G_{i-1})\right\}$. 

Each iteration of the algorithm is called a \emph{basic reduction}, denoted by $R_i$, which can be described by a pair $(v_i,G_{i-1})$. 
An \emph{execution} $\calE:=(R_1,\dots, R_k)$ of our greedy algorithm is the ordered sequence of basic reductions performed by the algorithm. 

To analyse an execution, we will only require local information for each basic reduction. Given a basic reduction $R_i=(v_i,G_{i-1})$, we call $v_i$ its 
\emph{root vertex}, its neighbors the \emph{middle vertices}, and together they form the \emph{ground} of the reduction, namely the set of vertices which are removed when the reduction is executed, written $ground(R_i)$. Vertices at distance two from the root are the \emph{contact vertices}. The set of contact vertices is denoted by $contact(R_i)$. Then, the edges between middle and contact vertices are called \emph{contact edges}.

From now on, we will consider that two basic reductions $R=(v,G)$ and $R'=(v',G')$ are isomorphic if there exists a one-to-one function $\phi:N_G[v]\longrightarrow N_{G'}[v']$ such that $\phi(v)=v'$, $u$ and $w$ are adjacent in $G$ if and only if $\phi(u)$ and $\phi(w)$ are adjacent in $G'$, and if each middle vertex $u$ is incident in $G$ to the same number of contact edges than $\phi(u)$ in $G'$.
Finally, the \emph{degree} of a basic reduction is defined as the degree of its root vertex. 

Figure \ref{fig:basicreductions} presents a table of all possible basic reductions of degree at most two in sub-cubic graphs. 
Notice that the middle vertices must have degrees equal to or greater than the degree of the reduction. 


\begin{figure}
    \centering
    \includegraphics[width=16cm]{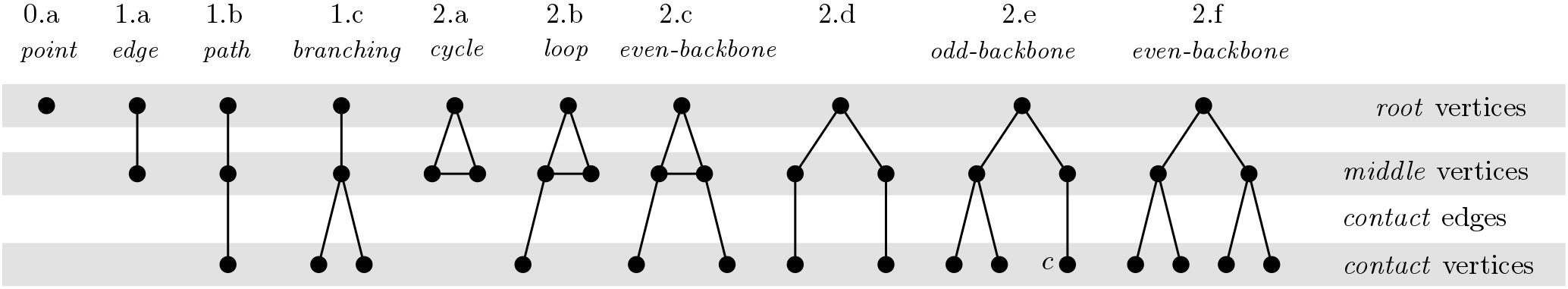}
    \caption{Basic reductions of degree at most two in sub-cubic graphs. We will refer later to these basic reductions by their names, for instance 2.b is a \emph{basic loop reduction}. In this picture, we have drawn contact vertices as distinct vertices, but in a reduction, several contact edges may be incident to the same contact vertex.
    When the right-most contact vertex $c$ of 2.e has degree three, this reduction is an odd-backbone reduction. Notice that in this case, the middle vertex of degree two is also the root of a basic odd-backbone reduction. } 
    \label{fig:basicreductions}
\end{figure}



\subsection{Potential function of reductions}\label{section:potential}

Suppose that we are given an independent set $I$ in a connected graph $G=(V,E)$ and an execution $\calE = (R_1,\dots,R_k)$ of a greedy algorithm on the input graph $G$. This execution is associated to a decreasing sequence of subgraphs of $G$:


$$
G=G_0 \supset \dots \supset G_{k}=\emptyset,
$$
where 
$G_i=G\left[V\setminus \bigcup_{j=1}^i ground(R_j)\right]$ is the induced sub-graph of $G$ on the set of vertices $V\setminus \bigcup_{j=1}^i ground(R_j)$.

Given a basic reduction $R_i=(v_i,G_{i-1})$, we define \emph{loan edges} of $R_i$ as all contact edges with a white contact vertex. Notice that the middle vertex of a loan edge can either be either black or white. The \emph{loan} of reduction $R_i$, denoted by $loan_{I}(R_i)$ corresponds to its total number of loan edges.

We also define the \emph{debt} of a \emph{white} vertex in the ground of $R_i$ as the number of times this vertex was incident to a loan edge, let us call it $e'$, of a reduction that was previously executed. Such loan edge $e'$ is also called a {\em debt} edge of reduction $R_i$.
  It turns out
that the debt of a white vertex corresponds exactly to the difference between its degree in the original graph $G$ and in the current graph $G_{i-1}$. Similarly, we define the debt of a reduction as the sum of the debts of the vertices of its ground.
$$
debt_{G,I}(R_i) = \sum_{u \in ground(R_i)\setminus I} \left( d_G(u)-d_{G_{i-1}}(u) \right)
$$



 

Given two parameters $\gamma,\sigma\ge 0$, we now define the \emph{exact potential} of a reduction $R_i$, for $1\le i \le k$, as 
$$
\Phi_{G,I}(R_i):=\gamma-\sigma\cdot|I\cap ground(R_i)|+
loan_{I}(R_i)-debt_{G,I}(R_i)
$$

The \emph{exact potential} of an execution $\calE = (R_1,\dots,R_k)$ is the sum of the exact potential of all reductions:
$$
\Phi_{G,I}(\calE)=\sum_{i=1}^k  \Phi_{G,I}(R_i)
$$

Since the independent set produced by the greedy algorithm is maximal and the total debt and the total loan are equal we obtain the following property.

\begin{restatable}{prop}{propsumpotential}\label{prop:sumpotential}
Given an execution $\calE = (R_1,\dots,R_k)$, we have: $\Phi_{G,I}(\calE)=\gamma k -\sigma|I|$.
\end{restatable}


\begin{proof}
By the definition of the exact potential, this can easily be seen by a simple counting argument.
\end{proof}

Suppose we want to analyse the approximation ratio of a greedy algorithm for a given class of graphs $\calG$. 
Then if we manage to find suitable values $\gamma,\sigma$ such that all possible reductions have non negative potential, then a direct corollary of Proposition \ref{prop:sumpotential} is that \Greedy is an $(\gamma/\sigma)$-approximation algorithm in $\calG$.



In order to measure the potential of each reduction, we now define a new potential, called simply \emph{potential} that is a lower bound on the exact potential. This lower bound is obtained by supposing that the debt of each white vertex is maximal, or equivalently, that its degree in the original graph was equal exactly to $\Delta$. 
$$
 debt_{I}(R_i) := \sum_{u \in ground(R_i)\setminus I} \left( \Delta-d_{G_{i-1}}(u) \right)\ge debt_{G,I}(R_i)
$$
Then we define the \emph{potential} of reduction $R_i$, which is now independent from the original graph, as
$$
\Phi_{I}(R_i):=\gamma-\sigma\cdot|I\cap ground(R_i)|+
loan_{I}(R_i)-debt_I(R_i)
$$
To evaluate the potential of a reduction, we do not need anymore to know the set of reductions previously executed but simply the structure of the graph formed by the vertices at distance two from the root, and also which vertices are black/white, which reduces to a relatively small number of cases. We define the \emph{potential of an execution} similarly. 
 Obviously, this
new potential is a lower bound on the exact potential defined previously, and more precisely, we have the following fact.


\begin{restatable}{claim}{claimcomparepot}
\label{claim:comparepot}
Let $G$ be a graph with maximum degree $\Delta$, $I$ an independent set in $G$ and $\calE$ an execution in $G$. Then, 
$$
\Phi_{G,I}(\calE)=\Phi_I(\calE)+\sum_{v\notin I} (\Delta-d_G(v)).
$$
\end{restatable}

\begin{proof}
By the definition of the two potentials, this can easily be seen by a simple counting argument.
\end{proof}

\subsection{Warm-up I: New proof of $\frac{\Delta +2}{3}$-ratio for greedy on degree-$\Delta$ graphs} \label{section-general-case}

Halld{\'o}rsson and Radhakrishnan \cite{Halldorsson1997,HalldorssonR1994} proved that for any graph with maximum degree $\Delta$, the basic greedy algorithm 
obtains a $\frac{\Delta+2}{3}$-approximation ratio. In here, we present an alternative proof for the same result, but using our payment scheme. Our proof will be simpler and shorter compared to the proof in \cite{Halldorsson1997}. 

  Let us use the potential from the previous section, with parameters $\gamma= (\Delta+b)\frac{\Delta+2}{3}$ and $\sigma=\Delta+b$ where $b=1$ if $\Delta \equiv 2 \pmod 3$, and $b=0$ otherwise. The choice of the value $b$ is simply to ensure that the potential value is integer. 
As we remark before, if we can prove that the potential of any reduction is non-negative, then the approximation ratio of \Greedy in graphs with maximum degree $\Delta$ is $\gamma/\sigma=(\Delta+2)/3$.


\begin{restatable}{lem}{lemmaapxdelta}\label{lemma:apxdelta}
Let $G$ be a graph with maximum degree $\Delta$. For any basic reduction $R$ and any independent set $I$ we have
\begin{equation*}
    \Phi_I(R) := (\Delta +b) \cdot \frac{\Delta+2}{3} - (\Delta+b) \cdot |I\cap ground(R)| + loan_I(R) -debt_I(R) \ge 0
\end{equation*}
where $b=1$ if $\Delta \equiv 2 \pmod 3$, and $b=0$ otherwise.
\end{restatable}


\begin{proof}

Let $\mathcal{R}$ be the set of all possible basic reductions, and let $I$ be a maximum independent set in the input graph.
We note that although there are many types of reductions in $\mathcal{R}$, their structure is highly regular. The idea of the proof is to find the worst type reduction and show that its potential is non-negative. Observe that, if we want to find a reduction $R^\ast$ to minimize the potential, $R^\ast = \arg\min_{R\in \mathcal{R}} \Phi_I(R)$, such reduction intuitively needs more debt edges and vertices in $I$ and less loan edges. Also, if $v^\ast$ is the root of reduction $R$, then for each $v \in V(R) \setminus \{v^\ast\}$, if $d_R(v^\ast) = k$, then $d_R(v)\geq k$, by the greedy rule.
  For any reduction $R$, let $i$ be the number of vertices in $I \cap ground(R)$ and let $\ell$ be the number of vertices in $ground(R) \setminus I$. We have the following formulas:
$$loan_I(R) \geq (i +\ell - 1 -\ell) \cdot i,$$
$$debt_I(R) \leq (\Delta - i -\ell +1 ) \cdot \ell.$$

We will justify these bounds now. Let $G'$ be the current graph just before $R$ is executed. Note first that the degree of the root of $R$ is $i+\ell-1$. The lower bound on $loan_I(R)$ depends on the vertices in $I$, by the definition. By the greedy order, for each of vertex $v \in I$, $d_{G'}(v) \geq i+\ell-1$. There are at most $\ell$ vertices not in $I$ which can be connected to $v$, thus, the total number of loan edges of $v$ is at least $(i+\ell-1-\ell)$, and we have $i$ such vertices. Note that in this argument we have possibly missed all loan edges that are contact edges of $R$ with both end vertices from $ground(R) \setminus I$.
  The upper
bound on $debt_I(R)$ depends on $\Delta$, the degree of the root vertex and the number vertices not in $I$. The number of debt edges is at most $\Delta - i - \ell + 1$, as otherwise it violates the greedy order, and we have $\ell$ vertices not in $I$.
\begin{align}
\Phi_I(R) &= \frac{\Delta+b}{3} \cdot (\Delta +2) -  (\Delta+b) |I \cap ground(R)| + loan_I(R) -debt_I(R)  \nonumber \\
   &\geq \frac{\Delta+b}{3}(\Delta +2) - (\Delta+b) i + (i - 1)i - (\Delta - i -\ell + 1) \ell \nonumber \\
   &= \ell^2 - (\Delta - i + 1) \ell +
   \frac{\Delta+b}{3}(\Delta +2) - (\Delta+b) i +(i-1)i \nonumber
\end{align}

Let $F(\Delta,i,\ell)  = \ell^2 - (\Delta - i + 1) \ell +
   \frac{\Delta+b}{3}(\Delta +2) - (\Delta+b) i +(i-1)i$. Then, the question now is to find the minimum value of $F(\Delta,i,\ell)$ with constrains $\Delta, i,\ell \in \mathcal{Z}^{+} \cup \{0\}$. 
 We will first
prove that 
$F(\Delta,i,\ell) \geq b/3 - b^2/3 - 1/3$ for any $\Delta, i,\ell \in \mathcal{R}^{+} \cup \{0\}$.
For any fixed $\Delta$ and $i$ let us treat the function $F(\Delta,i,\ell)$ as a function of $\ell$. We know that it is a parabola with the global minimum at point $\ell$ such that $\frac{\partial F}{\partial \ell} = 0$, which gives us that $\ell = (\Delta - i+1)/2$.
 Plugging $\ell = (\Delta - i+1)/2$ into $F(\Delta,i,\ell)$, we obtain the following function:
$$ F(\Delta,i,(\Delta - i + 1)/2) = F(\Delta,i)  =
-\frac{1}{4} (\Delta - i+1)^2 
+
   \frac{\Delta+b}{3}(\Delta +2) - (\Delta+b) i +(i-1)i =
$$
$$
= \frac{3}{4}i^2 
- (\Delta/2 + 1/2 + b)i +
   \frac{\Delta+b}{3}(\Delta +2)
-\frac{1}{4} \Delta^2 - \frac{1}{2}\Delta - \frac{1}{4}.
$$

Similarly as above for any fixed $\Delta$, we see that the function $F(\Delta,i)  =
\frac{3}{4}i^2 
- (\Delta/2 + 1/2 + b)i +
   \frac{\Delta+b}{3}(\Delta +2)
-\frac{1}{4} \Delta^2 - \frac{1}{2}\Delta - \frac{1}{4}$
 as a function of $i$ is a parabola with the global minimum for $i$
such that $\frac{\partial F}{\partial i} = 0$, which gives us that $i = \frac{2}{3}(\Delta/2 + 1/2 + b)$.
  Plugging $i = \frac{2}{3}(\Delta/2 + 1/2 + b)$
in $F(\Delta,i)$ we obtain the following:
$$
  F\left(\Delta, \frac{2}{3}\left(\Delta/2 + 1/2 + b\right)\right) = F(\Delta) = b/3 - b^2/3 - 1/3.
$$
From the above we have that $F(\Delta,i,\ell) \geq b/3 - b^2/3 - 1/3$ for any $\Delta, i,\ell \in \mathcal{R}^{+} \cup \{0\}$.

Now, let us observe that if $\Delta \equiv 0, 1 \pmod 3$, then $F(\Delta,i,\ell)$ with $b=0$ is an integer whenever $\Delta, i$ and $\ell$ are integers. This means that in those cases we have
$F(\Delta,i,\ell) \geq -1/3$ which implies that $F(\Delta,i,\ell) \geq 0$. In case when $\Delta \equiv 2 \pmod 3$, we have that $F(\Delta,i,\ell)$ with $b=1$ is an integer whenever $\Delta, i$ and $\ell$ are integers. This again means that in those cases $F(\Delta,i,\ell) \geq -1/3$, again implying $F(\Delta,i,\ell) \geq 0$.
\end{proof}

\begin{cor}[\cite{Halldorsson1997}]\label{main-theorem-upper-bound-delta}
 For \pr{MIS} on a graph with maximum degree $\Delta$, \Greedy achieves an approximation ratio of $\frac{\Delta +2}{3}$.
\end{cor}

This theorem implies only an approximation of $5/3$ for sub-cubic graphs. To do significantly better we need a stronger potential, better advice for greedy and a new method of analysis.

\subsection{Warm-up II: Proof of $\frac{\Delta+6}{4}$-ratio for   greedy  on degree-$\Delta$ triangle-free graphs}

Let us use the potential from Subsection \ref{section-general-case}, but with parameters $\gamma = \Delta \cdot \frac{\Delta+6}{4}$ and $\sigma = \Delta$. For any basic reduction $R$ and any independent set $I$ we have:
\begin{equation*}
    \Phi_I(R) = \Delta \cdot \frac{\Delta+6}{4}-\Delta \cdot |I \cap ground(R)|+\emph{loan}_I(R) -\emph{debt}_I(R).
\end{equation*}
 If we prove 
that $\Phi_I(R)\geq 0$ for any reduction $R$, and any independent set $I$, then the approximation ratio of \Greedy on triangle-free graphs with maximum degree $\Delta$ is $\gamma/\sigma=\frac{(\Delta+6)}{4}$. Let us first assume that the root vertex of $R$ is white. Then, in analogy to Subsection \ref{section-general-case} we have:
$$
\emph{loan}_I(R) \geq (i + \ell - 1 - 1 ) \cdot i,
$$
$$
\emph{debt}_I(R) \leq (\Delta - i -\ell +1 ) \cdot \ell. 
$$ Note that the upper bound on \emph{debt} edges is the same, however, the lower bound on \emph{loan} edges is significantly greater than the bound in Subsection \ref{section-general-case}. We will justify this first lower bound. Observe, that for a triangle-free graph, for any reduction $R$, no two middle vertices of $R$ are adjacent. Note first that the degree of the root of $R$ is $i+\ell-1$. We will obtain a lower bound on 
$\emph{loan}_I(R)$ by counting the number of contact edges incident on any middle vertex $v \in I$ of $R$.
By the greedy order, for each such vertex $v$, $d_{G'}(v) \geq i+\ell-1$, where $G'$ is the current graph. 
Because the roof of $R$ was assumed to be white, there is at most one vertex (the root vertex) not in $I$ which can be connected to $v$. Thus, the total number of loan edges of $v$ is at least $(i+\ell-1-1)$, and we have $i$ such vertices. Then we obtain further that:
\begin{align*}
    \Phi_I(R) & \geq \frac{\Delta}{4} (\Delta+6) - \Delta i + (i + \ell -2) i - (\Delta -i - \ell + 1) \ell \\
    & = \ell^2 - (\Delta -2i +1) \ell + \frac{\Delta}{4} (\Delta +6) - \Delta i + (i -2)i.
\end{align*}
Let $F(\Delta,i,\ell) = \ell^2 - (\Delta -2i +1) \ell + \frac{\Delta}{4} (\Delta +6) - \Delta i + (i -2)i$. Then, by using the same approach as in Subsection \ref{section-general-case} we obtain that $F(\Delta,i,\ell) \geq - i + \Delta - \frac{1}{4}$ for any $\Delta,i,\ell \in \mathcal{R}^+ \cup {0}$. 
%
%
Let now $d \leq \Delta$ be the degree of the root of $R$. If $d \geq i+1$, then $F(\Delta, i, \ell) \geq 0$ and $\Phi_I(R) > 0$. Suppose now that $d \leq i$, then obviously $d=i$. In such case, because $d=i+ \ell -1$, $\ell = 1$, thus:
\begin{align*}
    \Phi_I(R) &= \frac{\Delta}{4}(\Delta+6) - \Delta d + (d - 1)d \geq \frac{4\Delta - 1}{4} > 0,
\end{align*} where the first inequality follows by the fact the quadratic function is minimized for $d=(\Delta+1)/2$.

Let us now assume that the root of $R$ is black. Noting that $i=1$ and $\ell \leq \Delta$, we obtain:
\begin{align*}
 \Phi_I(R) & = \frac{\Delta}{4} \cdot (\Delta +6) - \Delta \cdot |I \cap ground(R)| + \emph{loan}_I(R) - \emph{debt}_I(R) \\
    & = \frac{\Delta}{4} \cdot (\Delta+6) - \Delta  + 0 - (\Delta - 1  - \ell + 1) \ell = \frac{\Delta^2}{4} + \frac{3}{2}\Delta - \Delta \ell + \ell^2 \geq \frac{3}{2} \Delta,
\end{align*} where the last inequality follows by the fact that the quadratic function is minimized for $\ell=\Delta/2$.
%
%
Thus we proved:
\begin{restatable}{theo}{theotrianglefree}\label{theo:trianglefree}
  For \pr{MIS} on a triangle-free graph with maximum degree $\Delta$, any greedy algorithm achieves an approximation ratio of $\frac{\Delta+6}{4}$.
\end{restatable}

\ignore{
\begin{align*}
    \Phi_I(R) &= \frac{\Delta+b}{4} \cdot (\Delta +6) - (\Delta +b) |I \cap ground(R)| + \emph{loan}_I(R) - \emph{debt}_I(R) \\
    & \geq \frac{\Delta+b}{4} \cdot (\Delta+6) - (\Delta +b)i + (i + \ell -2)i - (\Delta -i - \ell + 1)\ell \\
    & = \ell^2 - (\Delta -2i +1) \ell + \frac{\Delta+b}{4} \cdot (\Delta +6) - (\Delta +b)i + (i -2)i
\end{align*}
Some calculus..
let $\frac{\partial F}{\partial \ell} =0$, then $0 = 2\ell - (\Delta - 2i + 1)$, then $\ell = (\Delta - 2i + 1)/2$, then 

\begin{align*}
    F(plugging) & = (\Delta - 2i + 1)^{2}/4 -(\Delta - 2i + 1)^{2}/2+ \frac{\Delta+b}{4} \cdot (\Delta +6) - (\Delta +b)i + (i -2)i \\
    & = -bi - i + \frac{3}{2}b + \frac{\Delta b}{4} + \Delta - \frac{1}{4}
\end{align*}
Let $b=0$, then
}


\section{Subcubic graphs} \label{Section:greedy-subcubic}

The exact potential that we use for subcubic graphs is given by the values $\gamma=5$ and $\sigma=4$. 
The table in Figure \ref{fig:negativepot} shows the potential of several basic reductions for some different independent sets. Unfortunately, as one can see in Figure \ref{fig:negativepot}, there exists reductions with negative potential. The goal of our additional advice for greedy will be to deal with these cases. The first step is to collect some consecutive basic reductions into one \emph{extended reduction} so that the potential of some basic reductions is balanced by others. For instance, one way to deal with the basic reduction 2.d in Figure \ref{fig:basicreductions}, which can have potential $-2$ (see $(a)$ in Figure \ref{fig:negativepot}), is to force \Greedy to prioritize a vertex of degree two with a neighbor with degree three. Therefore, if at some point the reduction 2.d is executed it means that the current graph is a disjoint union of cycles. This allows us to consider that the whole cycle forms an extended reduction --- that we will call as \emph{cycle reduction} --- and we will see later that its potential is now at least $-1$. 
This \emph{advised} greedy algorithm, 
 called \textbf{MoreEdges} in  \cite{10.1007/BFb0015418}, improves the approximation ratio from $5/3$ to $3/2$ in sub-cubic graphs. 
 This result can easily be proved by using our potential function with parameters $\gamma,\sigma = 6,4$. Such approximation simply follows from the fact that all reductions have now non-negative potential.

\begin{figure}
    \centering
    \includegraphics[width=15cm]{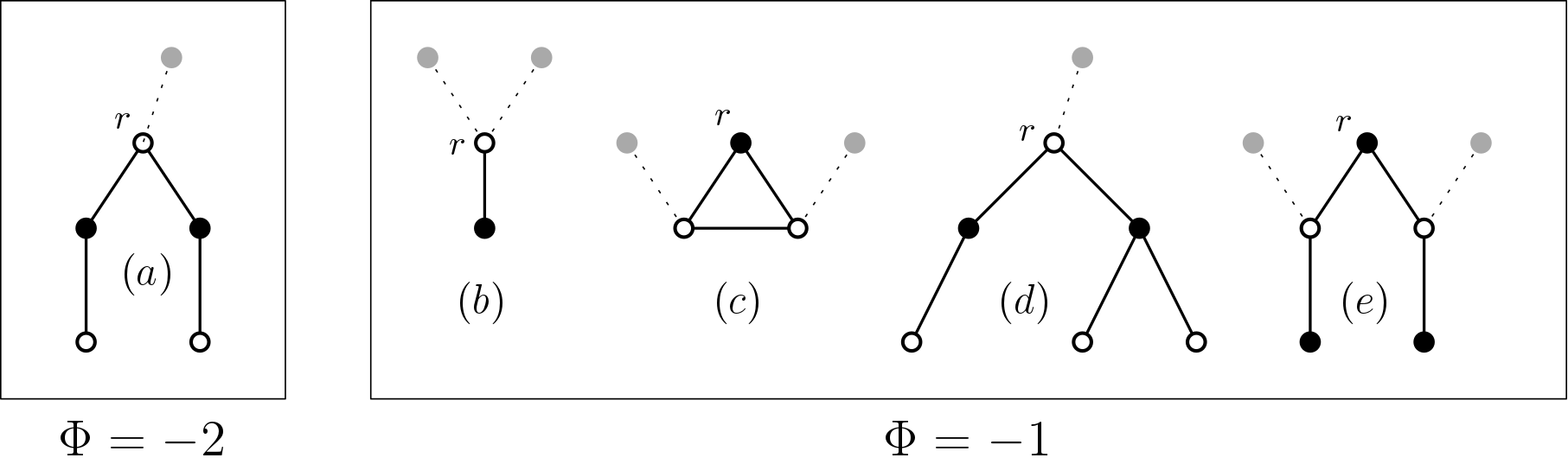}
    \caption{Basic reductions with negative potential. The root vertex of reductions is denoted by the letter $r$. Dotted edges translate the debt of each white vertex. Grey vertex can either be black or white.}
    \label{fig:negativepot}
\end{figure}

An useful observation in order to define an appropriate extended reduction is to notice that the (basic) path reduction 
(1.b from Figure \ref{fig:basicreductions}) has potential at least zero. 
This observation suggests to introduce the following notion. Given a graph $G=(V,E)$ we will say that the set $B=\{w,v_1,\dots,v_b,w'\}\subset V$ is a \emph{backbone} if the induced subgraph $G[\{v_1,\dots,v_b\}]$ is a path and if $w$ and $w'$ have both degree three. In this case, $w$ and $w'$ are called the \emph{end-points} of the backbone $B$. Moreover, when $b$ is odd (\resp even), we will say that $B$ is an \emph{even} (\resp \emph{odd}) backbone --- notice the asymmetry --- which corresponds to the parity of the number of \emph{edges} between the end-points. As an example, the ground of the basic even-backbone reductions (2.f and 2.c in Figure \ref{fig:basicreductions}) are special case of an even-length-backbone (of edge-length two).

\subsection{Extended reductions}

An \emph{extended reduction} $\overline{R}=(R_1,\dots,R_s)$ is a sequence of basic reductions $R_i$ of special type that we will precisely describe in the next paragraph. All different extended reductions are summarised in Proposition \ref{prop:listext}.  
To facilitate the discussion, when there is no risk of confusion, we will simply call it a \emph{reduction}. 
The \emph{size} of an extended reduction $\overline{R}$, written $\vert\overline{R} \vert$ is the number of executed basic reductions. Its \emph{ground} naturally corresponds to the union of the grounds of its basic reductions, $ground(\overline{R}):= \bigcup_i ground(R_i) $ and its \emph{root} is the same as the root of the first basic reductions. Finally, the \emph{contact vertices} corresponds to all contact vertices of its basic reductions that are not in $ground(\overline{R})$.
 The \emph{degree} of a reduction is the degree of the first executed basic reduction. 
    All basic 
reductions of Figure \ref{fig:basicreductions} except 2.d will be considered as (extended) reductions of size one. In particular, all (extended) reductions of degree one considered by the algorithm have size one. 
 Other considered (extended) reductions of degree two have a ground which is a backbone (except the case of odd-backbone where one end-point is excluded). When the two end-points of the backbone are the same vertex, it corresponds to a \emph{loop reduction}. Otherwise, reductions associated to an even and odd backbone are respectively called \emph{even-backbone} and \emph{odd-backbone} reductions. When these reductions have size at least two, they correspond to a sequence $\overline{R}=(R_1,\dots,R_s)$ of basic reductions where: 
 \begin{itemize}
     \item The first (basic) reduction $R_1$ is 2.e from Figure \ref{fig:basicreductions}. 
     \item Intermediate reductions $R_i$, with $2\le i \le s-1$, are basic path reduction (1.b from Figure \ref{fig:basicreductions}), where the root vertex of $R_i$ is the contact vertex of $R_{i-1}$. 
     \item The final (basic) reduction $R_s$ corresponds to:
     \begin{itemize}
        \item branching (1.c) or path (1.b), when $\overline{R}$ is an \emph{even-backbone} reduction. The case $R_s=$ \emph{path}, occurs when the end-points are adjacent.
        \item path (1.b), when $\overline{R}$ is an \emph{odd-backbone} reduction.
        \item point (0.a) or edge (1.a), when $\overline{R}$ is a \emph{loop} reduction, depending on the parity of the length of the backbone. Recall that the two end-points are identical in this case.
     \end{itemize}
 \end{itemize}
 
 \begin{figure}
     \centering
     \includegraphics[width=13cm]{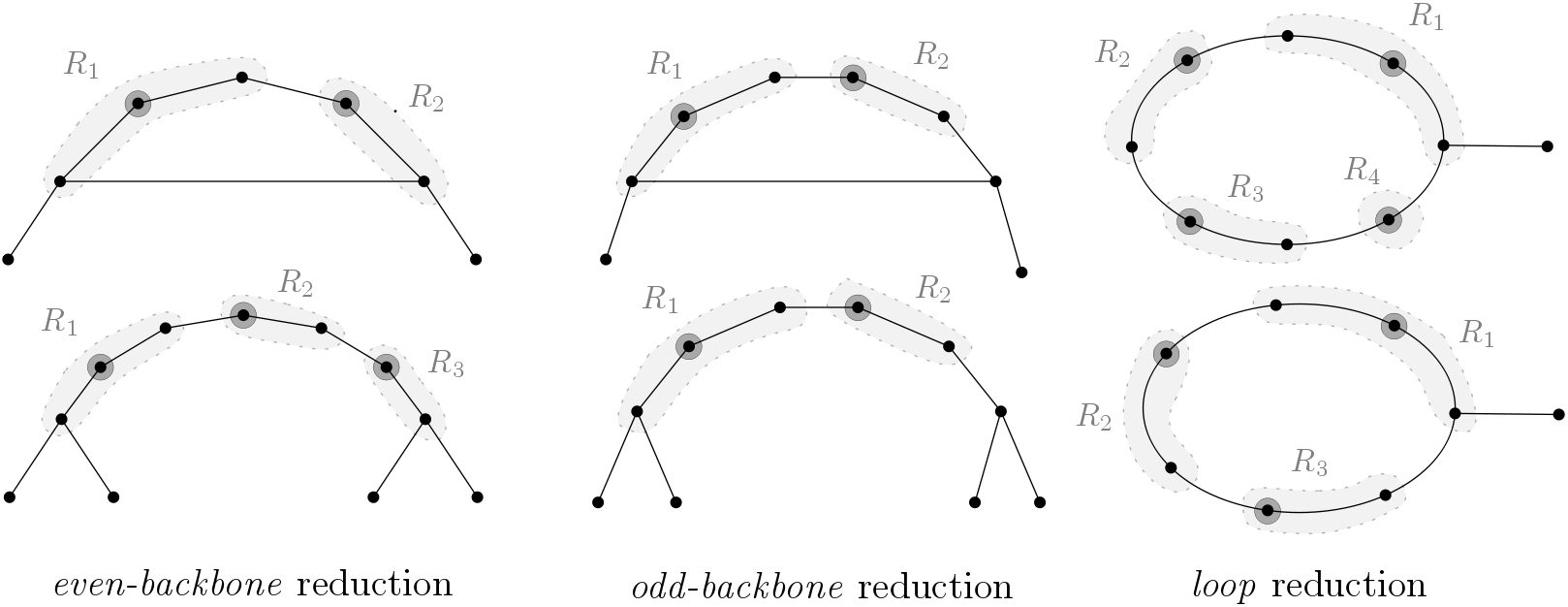}
     \caption{Some examples of (extended) reductions of degree two and size at least two. Light grey areas indicate the ground of each executed basic reduction, that together form the ground of the (extended) reduction. Vertices surrounded by grey rings are the roots of the corresponding basic reductions. }
     \label{fig:extendedred} 
 \end{figure}
 
  We give examples of different types of extended reductions of degree two in Figure \ref{fig:extendedred}. Some further remarks are in place here:
  \begin{itemize}
  \item The following basic reductions in Figure \ref{fig:basicreductions} are special case of (extended) reduction of size one.
  \begin{description}
     \item[2.a :] cycle reduction.
      \item[2.c and 2.f :] even-backbone reduction.
      \item[2.e :] odd-backbone reduction. This applies only when the right-most contact vertex $c$ has degree three.
      \item[2.b :] loop reduction.
\end{description}
      \item The root of an even-backbone, odd-backbone or loop reduction is always the neighbor of one of the end-points of the backbone. For even-backbone and loop reduction of even-length, any of the two choices leads to the same solution. In the case of an odd-loop reduction, the size of the solution --- and therefore also the potential --- and the ground of the reduction is exactly the same. For loop and even-backbone reductions, the ground of the reduction is the full backbone. However, for odd-backbone reductions, given one backbone, there are \emph{two distinct} possible root, associated to two distinct ground. For each odd-backbone reduction, only one end-point is contained in the ground. See Figure \ref{fig:extendedred}.
      \item All basic reductions of an extended reduction, except the first one have degree at most one, so that at any given moment, any executed basic reduction has minimum degree in the current graph. This means that we are allowed to execute the full extended reduction without violating our original greedy rule.
     \end{itemize}

In what follows, when we refer to an extended reduction in two different (and equivalent) ways. We will either write its name with the first capital letter or we will write its name with the first lower-case letter followed by the word ``reduction''. Thus, for example, we will say a loop reduction or just \textsf{Loop}, or an even-backbone reduction or just \textsf{Even-backbone}, etc. Note that basic reductions are special case of extended reductions, and therefore they also may follow this convention.

\subsection{Ultimate advices for \Greedy }

We now describe the additional rules used to reach the best possible approximation. This advised greedy algorithm will be called \Greedystar. The first of these rules is to execute basic reductions such that the obtained sequence can be grouped in a sequence of extended reductions as described above. This is justified by Proposition \ref{prop:listext}. This choice is always possible since all basic reductions from Figure \ref{fig:basicreductions}, except 2.d, are special cases of extended reductions. In the case where any minimum degree vertex is the root of a basic reduction 2.d, the graph must be a disjoint union of cycles. In this case we are able to execute \Greedystar so that its execution corresponds to a sequence of cycle reductions. This argument leads to Proposition \ref{prop:listext}. 

\begin{restatable}{prop}{proplistext}\label{prop:listext}
For each sub-cubic graph with minimum degree at most two, it is always possible to execute one of the following (extended) reductions:
\begin{center}
    \emph{\textsf{Point} - \textsf{Edge} - \textsf{Path} - \textsf{Branching} - \textsf{Loop} - \textsf{Cycle} - \textsf{Even-backbone} - \textsf{Odd-backbone}.}
\end{center}
\end{restatable}



\paragraph{\Greedystar order.}
When several choices of reductions are possible, \Greedystar will have to select one with the \emph{highest priority}, according to the following \emph{order} from the highest to lowest priority:
\begin{enumerate}
    \item \textsf{Point}, \textsf{Edge}, \textsf{Path}, \textsf{Branching}, \item \textsf{Cycle} or \textsf{Loop},
    \item \textsf{Even-backbone},
    \item {\sf Odd-backbone}.
\end{enumerate}

Any two reductions among the first group or any two reductions among the second group can be arbitrarily executed first, as soon as both have the minimum degree. We say that a reduction is a \emph{priority reduction} if there exists no reduction in the same graph with strictly higher priority. Thus a priority reduction is one of the highest priority reductions in the current graph. One implication of this order is that when an \textsf{Even-backbone} is executed, it means that the current graph does not contain any degree one vertices, or any loop reduction. Additionally, when the priority reduction is \textsf{Odd-backbone}, the graph does not contain any \textsf{Even-backbone}. These structural observations will be useful later.


When the priority reduction is an even-backbone or an odd-backbone reduction, \Greedystar applies the following two additional rules.





\paragraph{\textsf{Even-backbone} rule.}

Suppose that the priority reduction in the current graph $G$ is the even-backbone reduction, and several choices are possible.
Unfortunately, picking arbitrarily one of these reductions can lead to a solution with poor approximation ratio. For instance, consider the graph $H_{i}$ in Figure \ref{fig:lowerbound}, highly inspired by \cite{10.1007/978-3-540-27796-5_5}. 

It turns out that the difficulty comes from the fact that executing an even-backbone reduction can split the graph into several connected components, each of them having a negative potential. To address this issue, we want to make sure that we are able to ``control'' the potential of almost all of these connected components. This right choice, followed by \Greedystar, is given by the following lemma. For any reduction $\overline{R}$ in a graph $G$ we will 
say that $\overline{R}$ \emph{creates} connected components $H_1,\dots, H_s$ if they are the connected components of the graph $G[V\setminus ground(\overline{R})]$. Intuitively, it suffices to execute an even-backbone reduction $\overline{R}$ such that all other even-backbone reductions are all present in the same connected component created by $\overline{R}$.

\begin{figure}
    \centering
    \includegraphics[width=15cm]{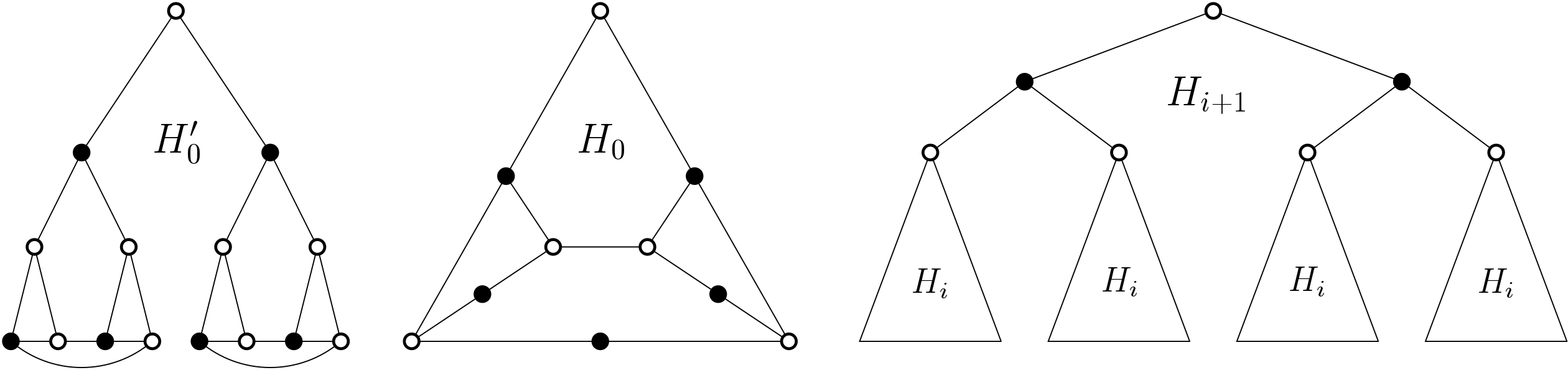}
    \caption{For any $i\ge 0$, the highest priority reduction in $H_i$ is \textsf{Even-backbone}, and picking recursively the top vertex leads to a solution where the approximation ratio tends to $17/13>5/4$ when $i$ tends to infinity, when $H_0$ is used as the base gadget. However, if we use $H_0'$ as the base gadget, because the greedy choice of the reduction is essentially unique at each stage, these instances show that {\em any} greedy algorithm has an approximation guarantee that tends to $5/4$ when $i$ tends to infinity. This second family is due to Halld{\'o}rsson and Yoshihara \cite{10.1007/BFb0015418}.}
   \label{fig:lowerbound}
\end{figure}

\begin{restatable}{lem}{lemrulefor}\label{lem:rulefor26}
Let $G$ be a connected graph, with no degree one vertices, and no loop reduction. Let $\calB=\{\overline{R}_1,\dots, \overline{R}_p\}$ be the set of all even-backbone reductions in $G$. Each even-backbone reduction $\overline{R}_i$ has two root vertices $r_i$ and $r_i'$. In the case when $\overline{R}_i$ has only one root, we set $r_i=r_i'$.  Then, there exists one even-backbone reduction, say $\overline{R}_1$ that satisfies the following property. Let $H_1,\dots, H_t$ be the connected components created by $\overline{R}_1$, with $1\le t\le4$. Then either $t=1$, or $t \geq 2$ and then the following is true.
If there exist $r_i, r_j$ for some $i, j \geq 2$ and $i \not = j$, such that $r_i \in V(H_1)$ and $r_j \in V(H_2)$ (in words: $r_i$ and $r_j$ belong to two different connected components among $H_1,\dots, H_t$), then at least one of $r_i, r_i', r_j, r_j'$ is a contact vertex of $\overline{R}_1$.
\end{restatable}

\begin{proof}(Lemma \ref{lem:rulefor26})
Let $a\in V(G)$ be any degree three vertex. 
Consider a graph $\widetilde{G}$ obtained from $G$ by replacing each backbone from $r_i$ to $r_i'$ by a single degree two vertex which is also called $r_i$.
On this \emph{contracted} graph, let $d_{\widetilde{G}}(u,v)$
denote the shortest path distance (\ie with minimum number of edges) between vertices $u,v \in V(\widetilde{G})$ in $\widetilde{G}$. Now, let us pick the root $r_i$ in $\widetilde{G}$ that has the largest distance $d_{\max} := \max_i d_{\widetilde{G}}(r_i,a)$ from $a$. Without loss of generality this is $r_1$ : $d_{\widetilde{G}}(r_1,a)=d_{\max}$. Denote by $H_1,\dots, H_t$ the connected components created after executing the corresponding even-backbone reduction $\overline{R}_1$. At most one connected component, say $H_1$, contains $a$. Suppose that there is another connected component, i.e., $H_2$, that contains a vertex $r_j$. Any path from $r_j$ to $a$ intersects $ground(\overline{R}_1)$, including the shortest one, and $d_{\widetilde{G}}(r_j,a)=d_{\widetilde{G}}(r_1,a)=d_{\max}$. It follows that $r_j$ and $r_1$ have one common neighbor $b$, so that $d_{\widetilde{G}}(r_1,r_j)=2$. In particular, in the original graph $G$, $r_1$ (or $r_1'$) is at distance two from $r_j$ (or $r'_j$) which means that this vertex is a contact vertex of $\overline{R}_1$. 
\end{proof}

Notice that this proof is constructive and allows us to find the appropriate \textsf{Even-backbone} in time $\calO(\vert V\vert)$.



\paragraph{\textsf{Odd-backbone} rule.} 

Suppose now that the priority reduction is the odd-backbone reduction. In this case, \Greedystar chooses the one that was created latest. More formally, suppose that we are given a partial execution $\overline{R}_1, \dots, \overline{R}_j$ in a graph $G$ such that the priority reduction in $G_j$ is an odd-backbone reduction, where $G_i$ is the subgraph of $G$ obtained from $G$ after the execution of $\overline{R}_1, \dots, \overline{R}_i$, for $i=1,\ldots,j$. We associate to each vertex $v$ of degree two a \emph{creation time} $t_v\in \{0,\dots, j-1\}$, such that $t_v$ is the greatest integer such that $v$ had degree three in $G_{t_v - 1}$. Moreover, if $v$ had already degree two in the original graph $G$, then set $t_v=0$. When $t_v\ge1$, this means that $v$ was a contact vertex of $t_v$-th reduction. Then, the \emph{creation time} of an (odd) backbone $B$ is the greatest creation time over all vertices of degree two in $B$.

\Greedystar picks the odd-backbone reduction that has the greatest creation time, among all possible odd-backbone reductions. If several choices are possible, it can pick any of them.


We believe that this rule is not necessary in the sense that it does not improve the approximation ratio. However, this rule makes the algorithm easier to analyse. Intuitively, with this rule, we can not have several successive reductions with negative potential within the same connected component. 

\paragraph{Rule for cubic graphs.} When the input graph is cubic, \ie each vertex has degree exactly three, then the first reduction has degree three. However, this is the only degree three reduction executed during the whole execution since there will always be a vertex with degree at most two. In such a situation, we guess the first degree three vertex to pick, so that the potential of the associated execution is positive. By guessing, we mean choosing {\em any single fixed} vertex $u$ and then trying all \emph{four} executions, each starting from a vertex in the closed neighborhood of vertex $u$. We show later that the first step can only increase the total potential of the whole sequence. After this step, all reductions have degree at most two, and therefore in the following sections, we will always consider graphs with at least one vertex of degree at most two.


\begin{algorithm}[H]
\SetAlgoLined

 \KwIn{a graph $G$}
 \KwOut{an independent set $S$ in $G$}
 
 $S\leftarrow\emptyset$\;
 
 \If{all vertices have degree three}{
 Let $u$ be any vertex.\;
 
 Execute \emph{four} times the while loop (line \ref{line:while}) starting with $S=\{v\}$ and $G=G\setminus N_G[v]$, for all $v\in N_G[u]$ and output the maximum size solution.\;
    }
    
    \While{$G\neq \emptyset$}{\label{line:while}
       \If{the \emph{priority reduction} (according to \Greedystar order) in $G$ is \textsf{Even-backbone}}{
          Let $\overline{R}$ be the \textsf{Even-backbone} given by Lemma \ref{lem:rulefor26} (\textsf{Even-backbone} rule).\;
        }
        \If{the \emph{priority reduction} in $G$ is \textsf{Odd-backbone}}{
            Let $\overline{R}$ be the latest created \textsf{Odd-backbone} (\textsf{Odd-backbone} rule).\;
        }
        \Else{
            Let $\overline{R}$ be any priority reduction.\;
        }
        Let $v_1,\dots,v_s$ denote the roots of the basic reductions of $\overline{R}$.\;
        
        $S\leftarrow S\cup \{v_1,\dots,v_s\}$\;
        
        $G\leftarrow G \setminus ground(\overline{R})$
     
    }

 \Return S;
 \caption{ - \Greedystar}
 \label{algo:greedy}
\end{algorithm}

It is clear that the set returned by Algorithm \ref{algo:greedy} is an independent set. In the next section, we establish its approximation ratio.


\subsection{Analysis of the approximation ratio}

\begin{restatable}{theo}{theoratio}\label{theo:ratio}
\Greedystar is a $5/4$-approximation algorithm for \pr{MIS} in subcubic graphs.
\end{restatable}


Let $\calE=(\overline{R_1},\dots,\overline{R_\ell})$ is a sequence of extended reductions performed by \Greedystar on an input graph $G$. In order to analyse the approximation ratio of \Greedystar, we use our potential function in sub-cubic graphs ($\Delta=3$) with parameters $\gamma,\sigma= 5,4$. Given an independent set $I$ in $G$, the \emph{potential} of an (extended) reduction $\overline{R}$ is
\begin{equation*}
    \Phi_I(\overline{R}) = 5 \cdot |\overline{R}| - 4 \cdot |I\cap ground(R)| + \emph{loan}_I(R) -\emph{debt}_I(R).
\end{equation*}
  We start by
looking at the potential of (extended) reductions.



\subsubsection{Potential of extended reductions} 

\begin{restatable}{claim}{claimpotentialextended}\label{claim:potentialextended}
For any independent set $I$ we have the following potential estimates for the reductions:
\begin{center}
      \begin{tabular}{r|c|l}
    $\Phi_I(\textsf{Edge})\ge -1$  & $\Phi_I($\textsf{Path}$)\ge 0$  & $\Phi_I(\textsf{Point})\ge 1$ \\
       $\Phi_I(\textsf{Cycle})\ge -1$  & $\Phi_I(\textsf{Loop})\ge 0$  & $\Phi_I(\textsf{Branching})\ge 1$ \\
       $\Phi_I($\textsf{Odd-backbone}$)\ge -1$  &  $\Phi_I($\textsf{Even-backbone}$)\ge 0$ & 
  \end{tabular}
  \end{center}
\end{restatable}



For basic reductions, one can easily check by hand all possible cases. Notice that the worst case potential always arises when $I$ is maximum in the ground of a reduction. Figure \ref{fig:negativepot} presents these worst cases for basic reductions: \textsf{Edge}, \textsf{Cycle} and \textsf{Odd-backbone}. Figure \ref{fig:01pot} shows the worst potential cases of the remaining basic reductions: \textsf{Path}, \textsf{Even-backbone}, \textsf{Loop}, \textsf{Point}, and \textsf{Branching}. From the worst case potential of basic reductions, we can lower-bound the potential of (extended) reductions.

\begin{proof}(Claim \ref{claim:potentialextended})
It remains to prove lower-bounds for reductions of arbitrary size. See Figure \ref{fig:cycles}. An odd-backbone reduction is a sequence of basic reductions starting with 2.e (in Figure \ref{fig:basicreductions}), which has a potential at least $-1$ (Figure \ref{fig:negativepot} (d)), and a certain number of path reductions, with potential at least zero (Figure \ref{fig:01pot} (a) and (b)), so that the total potential is at least $-1$. More generally, the potential of an extended reduction is lower-bounded by the sum of the potentials of the first and the last basic reduction. For \textsf{Even-backbone} with non-adjacent backbone end-points, these first and last basic reductions are 2.e ($\Phi\ge-1$) and \textsf{Branching} ($\Phi\ge 1$, see Figure \ref{fig:01pot} (g) and (h))  so that the sum is non-negative. 

Consider now a cycle reduction $\overline{R}$ of length $n\ge 3$. Let us denote by $b$ and $w$, respectively, the number of black and white vertices, \ie $b=\vert I\cap ground(\overline{R}) \vert$ and $b+w=n$. Since \Greedy is optimal in degree at most two graphs and the size of $I$ is at most $\left \lfloor \dfrac{n}{2}\right\rfloor$, we have:

\begin{equation}\label{eq:cycle}
    \Phi_I(\overline{R})=5\left\lfloor \frac{n}{2} \right\rfloor - 4b - w = 5\left\lfloor \frac{n}{2} \right\rfloor - 3b - n
\ge 5\left\lfloor \frac{n}{2} \right\rfloor - 3\left\lfloor \frac{n}{2} \right\rfloor-  n  \ge -1.
\end{equation}

For \textsf{Loop} and \textsf{Even-backbone} with adjacent end-points, simply observe that their ground is a cycle with one or two additional edges. Each of these edges is either a debt edge --- the loan increases by one --- or the corresponding middle white vertex has now degree three --- the debt decreases by one. In any case, the potential increases by the number of added edges, so that we proved what we wanted.
\end{proof}

\begin{figure}
    \centering
    \includegraphics[width=12cm]{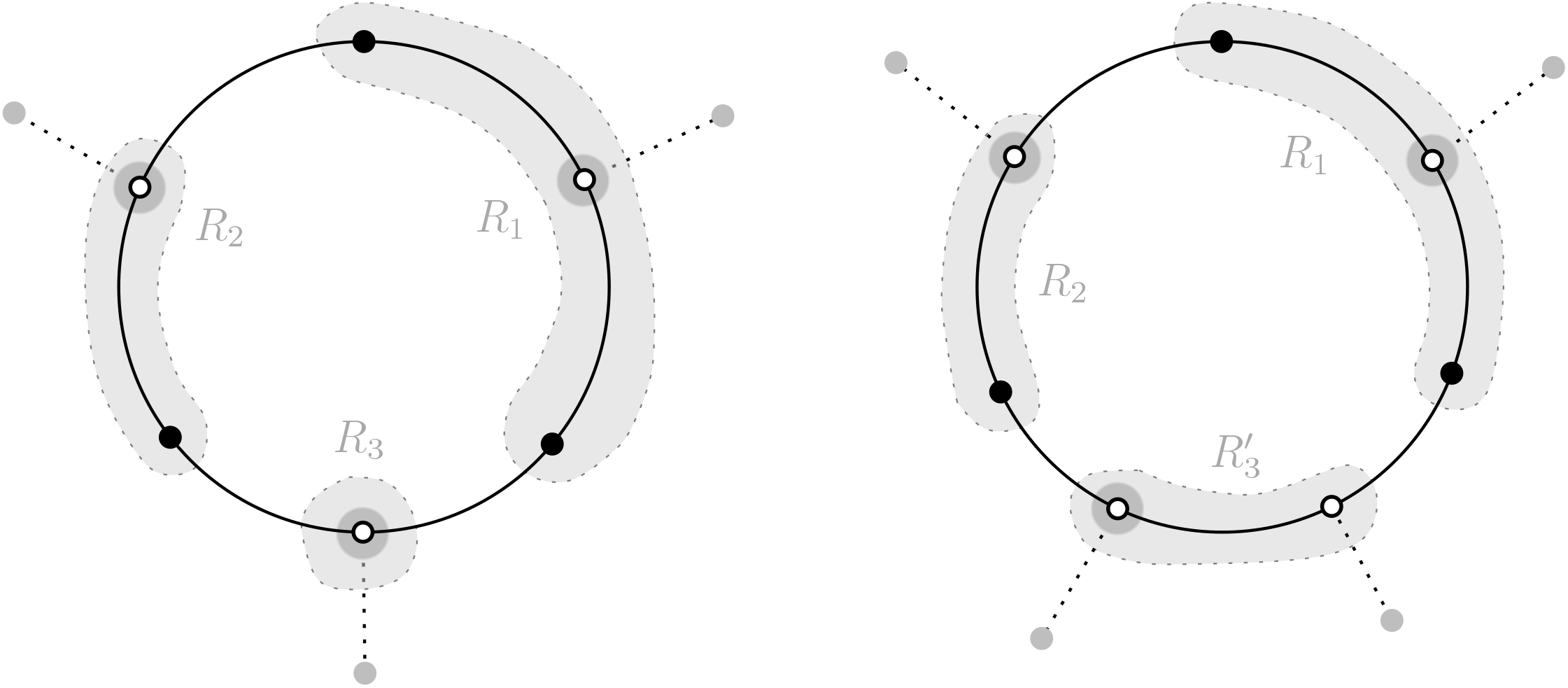}
    \caption{A \emph{cycle} reduction. On the left, a even cycle with potential $\Phi_I(R_1,R_2,R_3)= -2+0+2=0
    $ and on the right a odd cycle with potential $-1$. Vertices surrounded by a grey disk are the ones picked by the algorithm and dotted edges are debt edges. }
    \label{fig:cycles}
\end{figure}

Notice that the worst potential of \textsf{Cycle} and \textsf{Loop} depends on the length of the ground and the worst case corresponds to odd-length cycles. Moreover, notice that when its two end-points are adjacent, the potential of an \textsf{Even-backbone} is at least one.




\begin{figure}
    \centering
    \includegraphics[width=16cm]{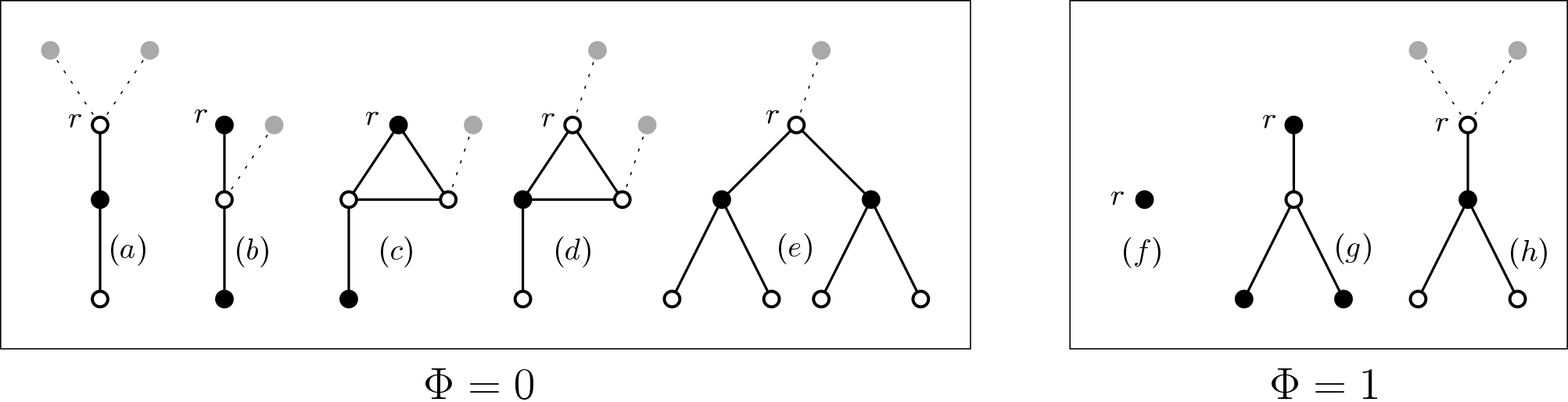}
    \caption{Basic reductions with worst potential equal to $0$ or $1$. Dotted edges translate the debt of each white vertex. Grey vertex can either be black or white.}
    \label{fig:01pot}
\end{figure}

\subsubsection{Proof of the inductive lemma}



As we have seen before, we are able to avoid reduction with potential $-2$ (Figure \ref{fig:negativepot} (a)) by grouping this basic reduction with the following ones, so that the resulting (extended) reduction, a cycle reduction, has now only potential $-1$. Unfortunately, we can not use the same trick to avoid reductions with potential $-1$. Moreover, such reductions: edge, (odd) cycle and odd-backbone reductions can not be avoided if we want to respect the original greedy constraint which is to pick a vertex with minimum degree. 

In order to prove that \Greedystar delivers a $5/4$-approximation, we show that the \emph{exact potential} of any execution is non-negative. Unfortunately, it is not true that the \emph{potential} of any execution is non-negative. For instance, picking the top vertex of $H_0$ in Figure \ref{fig:lowerbound} produces an execution with potential $-1$. Hopefully, we will prove that $-1$ is the worst value for the potential of any execution.

A potential problem may arise when an execution creates a lot of connected components where each corresponding execution has negative potential. This could possibly lead to an execution with arbitrarily negative potential. Such connected components might be created by reductions having many contact vertices, such as the even-backbone reduction. Our \textsf{Even-backbone} rule was designed to keep the potential of the created connected components under control, ensuring that at most one (or two) of them have negative potential. 

This suggests that to solve our problem we could try to characterize the type of graphs that can have negative potential, \ie, for which there exists an execution with negative potential. Finding such a characterization seems difficult but hopefully, since we know which reduction's potential is $-1$, we are able to formulate a necessary condition together with a suitable induction hypothesis. In the following, we describe this class of graphs called \emph{potentially problematic} graphs. Additionally, all executions starting with a negative potential reduction, namely the \emph{bad} odd-backbone reduction, can possibly have a negative potential. Notice that 'potentially' refers to the potential function and at the same time to the fact that there exist some executions on some potentially problematic graphs with non-negative potential.

Given a graph $G$ and an independent set $I$ in $G$, recall that a vertex $v \in V(G) \cap I$ from $I$ is {\em black} and {\em white} otherwise. We say that a backbone $B=\{w, v_1,\dots, v_b, w'\}$ is {\em alternating for} $I$ (or simply {\em alternating}) when $v_i$ is black (or white) if and only if $i$ is even. See Figure \ref{fig:helpproof} for an example of an odd alternating backbone. Notice that there is no restriction on the types of the end-points of the backbone.



The next definition of the white and black type reductions is absolutely crucial to our proof. The main intention of this definition is that the black and white reductions should have the greedy priority strictly higher than that of an odd-backbone reduction, and in many cases, also higher than that of an even-backbone reduction. However we must be very careful here because what will matter will be in some sense a ``parity'' of the reductions. Namely, the first observation is that the potential of an odd-backbone reduction, let us call it $\ol{R}$, can be $-1$ in the case when all its contact vertices are white (see Figure \ref{fig:negativepot} (d)). Therefore when such a reduction $\ol{R}$ is executed first, then we need to ``pay'' for it by showing that the potential of the following reductions in the connected components it creates will in total be zero (as this is the only way of keeping the total value of the potential to be at least $-1$). What we now want is that if a potential of a connected component $H$ created by $\ol{R}$ is negative, then because $\ol{R}$ has only white contact vertices in $H$, even {\em before} $\ol{R}$ is executed, $H$ must contain a reduction with higher priority than $\ol{R}$, thus leading to a contradiction. Our definition below will ensure this by guaranteeing that such a reduction in $H$ exists with a black root vertex 
in $H$ by $\ol{R}$'s contact edges).

A kind of a ``dual'' such high priority reduction is also needed in $H$. Imagine namely that the first executed reduction $\ol{R}$, that created component $H$, is also an odd-backbone reduction but its contact vertices are all black. Then they will ``block'' the black vertices in $H$. But then our definition below still guarantees an existence of a reduction in $H$, before executing $\ol{R}$, which has a white root vertex 
and has priority strictly higher than $\ol{R}$.

The third possibility is when $\ol{R}$ is an odd-backbone reduction with a white and black contact vertex. This is not a problem because the potential of such a reduction is non-negative or even positive. 
It turns out that we can deal with $\ol{R}$ being an even-backbone reduction in a different way.

There are some further technicalities and details and they can be read in the details of our full proof.



\begin{restatable}[Black and white type reductions]{df}{dftype}\label{df:type}
Given a graph $G$ and an independent set $I$ in $G$, we define the black or white type reductions in $G$ by the following rules:\\ {\bf(1)} Any path reduction or branching reduction in $G$ with black root and white middle vertex (\resp with white root vertex and black middle vertex) is a black (\resp white) type reduction.\\ 
{\bf(2)} Any Loop reduction $\overline{R}$ in $G$, where $I \cap ground(\overline{R})$ is a maximum independent set in $ground(\overline{R})$, whose both root vertices are white (\resp whose at least one root vertex is black) is called a white (\resp black) type reduction.\\
{\bf(3)} Any even-backbone reduction $\overline{R}$ in $G$, with an alternating backbone, whose both root vertices are white (\resp black) is called a white (\resp black) type reduction.
\end{restatable}



We note here that the black and white reductions correspond to the worst case of the potential (see Figures \ref{fig:negativepot} and \ref{fig:01pot}).
We also say that an an odd-backbone reduction $\overline{R}$ is {\em bad for $I$} (or simply {\em bad}) if $\Phi_I(\overline{R}) = -1$, see Figure \ref{fig:negativepot} (d). More generally, given an independent set $I$, we will say that a reduction $\ol{R}$ is \emph{bad} when its potential is minimized by $I$, \ie $\Phi_I(\ol{R})=\min\{\Phi_{I'}(\ol{R}), I' \text{ independent set}\}$. 

Notice that an odd-backbone reduction does not appear in the definition of the black and white type reductions. Recall that our intention was that any black/white reduction has a strictly higher priority than \textsf{Odd-backbone}.  


\begin{restatable}{df}{dfproblematic}\label{df:problematic}
Let $G$ be a connected graph with minimum degree at most $2$ and $I$ an independent set in $G$. We say that $G$ is \emph{potentially problematic for} $I$ (or just \emph{potentially problematic}), if either:\\
{\bf{(1)}} $G$ is an odd cycle or an edge and $I$ is maximum in $G$. \\
{\bf{(2)}}  or, there exists one reduction of black type \emph{and} one reduction of white type in $G$.
\end{restatable}



Notice that \textsf{Cycle} (and also \textsf{Edge}) reduction has potential $-1$ if and only if its ground has odd-length and $I$ is maximum (Claim \ref{claim:base-case}). 
In this situation, we will call these graphs \emph{bad} cycle and \emph{bad} edge.

The following Lemma states that any execution has a potential always at least $-1$. We prove this result by induction together with a necessary condition characterising executions with minimum potential. 

\begin{restatable}[Inductive low-debt Lemma]{lem}{lemmainduction}\label{lemma:induction}
Let $G$ be a connected graph with minimum degree at most two, $I$ an independent set in $G$ and $\calE=(\overline{R}_1,\dots,\overline{R}_\kappa)$ an execution on $G$. Then 
\vspace{-5pt}
\begin{enumerate}[(1)]
\itemsep0em 
\item $\Phi_I(\calE)\ge -1$ \label{-1}
\item \label{prob} If $\Phi_I(\calE) = -1$, then either

   \begin{enumerate}[(a)]
    \vspace{-5pt}
    \itemsep0em 
        \item \label{prob25} The first reduction $\overline{R}_1$ is a bad odd-backbone reduction.
        \item \label{prob26} or, $G$ is potentially problematic.
    \end{enumerate}
 
\end{enumerate}
\end{restatable}

 

Along the proof of Lemma \ref{lemma:induction} we will refer to several technical claims proved in Section \ref{section:technical} about extended reductions.

\begin{proof}(of Lemma \ref{lemma:induction}) 
%
We prove this result by induction on the number $\kappa$ of extended reductions in the execution $\calE$.
  First, if 
$\kappa=1$, then since $G$ is connected, the reduction is a terminal reduction, \ie point, edge or cycle reduction, and by Claim \ref{claim:potentialextended} we know that their potential is at least $-1$.  Moreover, if the potential is exactly $-1$, it is not difficult to see from the proof of Claim \ref{claim:potentialextended} that the reduction is either an edge or odd cycle reduction, which are potentially problematic graphs. For a detailed proof of this fact see Claim \ref{claim:base-case}.

Suppose now that $\calE$ consists of $\kappa\ge 2$ extended reductions. We will treat all cases depending on the first extended reduction $\overline{R}_1$. We denote its root and the contact vertices by letters $r$ and $c_i$, with $1\le i \le 4$. In all these cases we will apply the induction hypothesis to each connected component of the graph after executing the first reduction. 
Recall that we say that $\overline{R}_1$ \emph{creates} connected components $H_1,\dots, H_s$ if they are connected components of the graph $G[V\setminus ground(\overline{R}_1)]$. Here, $s$ with $1\le s\le 4$, denotes the number of connected components created by $\overline{R}_1$. Reductions executed by \Greedystar in distinct connected components are independent. Therefore, without loss of generality we can assume that each execution on $H_i$ corresponds to a sub-execution $\calE_i$ of $\calE$ so that $\calE=(\overline{R}_1,\calE_1,\dots,\calE_s)$. Notice that according to Proposition \ref{prop:sumpotential}, we have
\begin{equation*}\label{eq:sum}
    \Phi_I(\calE)=\Phi_I(\overline{R}_1)+\Phi_I(\calE_1)+\dots+\Phi_I(\calE_s).
\end{equation*}

 We first prove hypothesis (\ref{-1}) and (\ref{prob}) if the first reduction is \textsf{Path}, \textsf{Branching} or \textsf{Loop}.  
 
 \paragraph{When $\overline{R}_1$ is \textsf{Path}, \textsf{Branching} or \textsf{Loop}.}~These are easy cases since the number of connected components created (and therefore the potential of each corresponding execution) is always balanced by the potential of the reduction. This is precisely written in the following fact that can be easily verified thanks to Figure \ref{fig:basicreductions} and Claim \ref{claim:potentialextended}.

\begin{restatable}{obs}{obscontactpotential}\label{obs:contact-potential}
For any independent set $I$, and any reduction $\overline{R}$ that is \textsf{Path}, \textsf{Branching} or \textsf{Loop}:
$$
\Phi_I(\overline{R})\ge|contact(\overline{R})|-1.
$$
\end{restatable}



\noindent
{\bf Remark:} The inequality is also valid for an even-backbone reduction whose both end-points are adjacent.

Then, if $\overline{R}_1$ is a path, branching, or loop reduction then the number of connected components is at most the number of contact-edges, therefore by the induction hypothesis (\ref{-1}) on each connected component of $G[V\setminus ground(\overline{R}_1)]$, we have 
$$
\Phi_I(\calE)= \Phi_I(\overline{R}_1)+\Phi_I(\calE_1)+\dots+\Phi_I(\calE_s) \ge \Phi_I(\overline{R}_1)+(-1)s\ge \Phi_I(\overline{R}_1)+(-1)|contact(\overline{R}_1)|\ge -1.
$$
Moreover, if $\Phi_I(\calE)=-1$, then all these inequalities are tight and in particular, the potential of $\overline{R}_1$ is minimum. This implies that $\overline{R}_1$ must be a reduction of black or white type (\ref{techclaim:easyred}), and additionally, it must create exactly $|contact(R)|$ connected components, and each one must have potential $-1$. Applying the inductive assumption (\ref{prob}) to these connected components, together with the following Claim \ref{claim:induction} implies property (\ref{prob26}) for $G$. 


\begin{restatable}{claim}{claiminduction}\label{claim:induction}
Let $G$ be a connected graph, and $I$ an independent set in $G$. Consider $\calE=(\overline{R}_1,\dots,\overline{R}_\kappa)$ an execution in $G$. Let $H$ be a connected component created by the first reduction $\overline{R}_1$, such that all contact vertices of $\overline{R}_1$ that are in $H$ are all {white}  (\resp all black). Then, if $H$ is potentially problematic \emph{or} if the first reduction executed in $H$ is a bad odd-backbone reduction, then there exists a black (\resp white) reduction $\overline{R}$ in $G$ such that $ground(\overline{R})\subseteq V(H)$.
\end{restatable}

We will prove that in this situation, $\ol{R}_1$ can not be an \textsf{Odd-backbone}, since this would not be the priority reduction according to \Greedystar order.
This claim is useful in the sense that, when the potential of $\overline{R}_1$ is minimum then $\overline{R}_1$ is white (or black) and its contact vertices are all white (or all black) (Claims \ref{claim:base-case},\ref{techclaim:eb}), so that $G$ is a potentially problematic graph. 


\begin{proof}(Claim \ref{claim:induction})
Suppose that $H$ is created by $\overline{R}_1$ with contact vertices $c_1,\dots, c_t\in V(H)$ all white (\resp all black), with $1\le t\le 4$. 
We show that there exists one black (\resp white) reduction  $\overline{R}$ in $G$ such that $ground(\overline{R})\subseteq V(H)$.

First, assume that the first reduction executed in $H$, say $\overline{R}_2$, is a bad odd-backbone reduction. Denote by $B = \{w,v_1,\dots, v_{2b},w' \}$ its backbone, with end-points $w$ and $w'$. Since $\Phi_I(\overline{R}_2)=-1$, its backbone is alternating (Claim \ref{techclaim:od}) 
and without loss of generality we assume that $v_j$ is white if and only if $j$ is odd.

According to the \textsf{Odd-backbone}-rule, at least one of the contact vertices $c_i$ must be one vertex of its backbone. 
Since all $c_i$ are white (\resp black), the distance along $B$ between any pair of $\{c_1,\dots, c_t,w'\}\cap B$ (\resp $\{w, c_1,\dots, c_t\}\cap B$) is even. Therefore, there is an alternating even-length-backbone between any two consecutive ones. This implies the existence of a black (\resp white) even-backbone reduction in $G$, see Figure \ref{fig:helpproof}. 

\begin{figure}
    \centering
    \includegraphics[width=12cm]{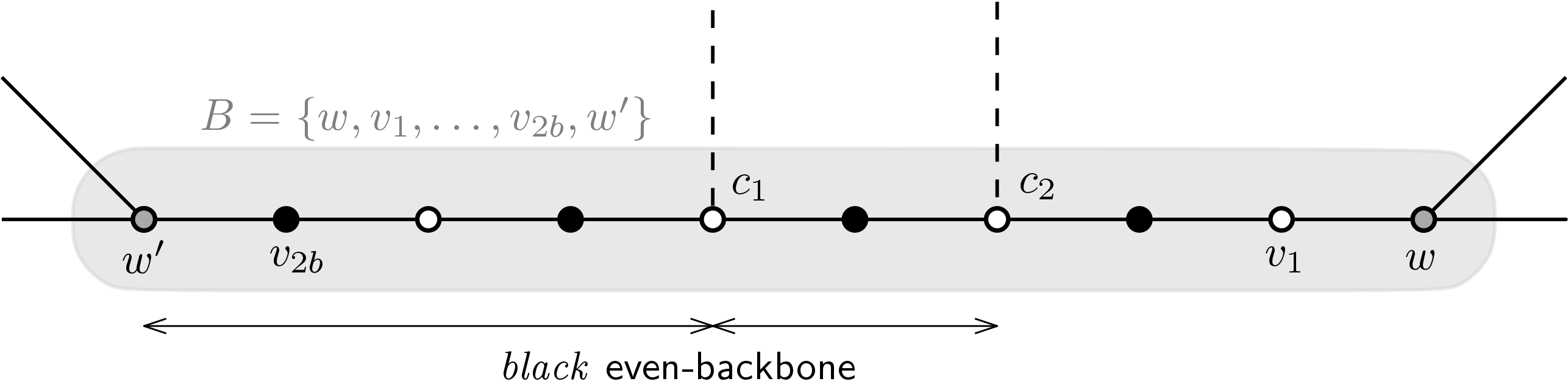}
    \caption{Existence of a black reduction in proof of Claim \ref{claim:induction}. $B$ is an alternating odd backbone in $H$. Grey end-points illustrate the fact that these vertices can either be black or white. Dashed edges are contact edges from $\overline{R}_1$. Before the execution, there exists a black (alternating) \textsf{Even-backbone} in $G$ between $w', c_1$ and $c_1,c_2$.}
    \label{fig:helpproof}
\end{figure}


Now, assume that $H$ is a potentially problematic graph, and suppose first that 
there is one black (\resp white) reduction $\ol{R}_2$ with a black (\resp white) root vertex $r$  in $H$. Notice that this vertex is distinct than all $c_i$, and then has the same degree in $G$. Then, if $d_H(r)=d_G(r)=1$, then $r$ is also the root of a black (\resp white) reduction in $G$. 
Then, suppose that $d_H(r)=d_G(r)=2$. 

If $r$ is the root of a black (\resp white) even-backbone reduction $\overline{R}_2$ in $H$, with backbone end-points $w$ and $w'$, then any two consecutive vertices of the set $\{w,w', c_1,\dots, c_t\}$ along this backbone form a alternating even-length-backbone. In particular, $r$ is the root of a black (\resp white) \textsf{Even-backbone} in $G$.

The case when $\ol{R}_2$ is a \textsf{Loop} is very similar, but slightly more technical. First, if there is no $c_i$ in the ground of $\ol{R}_2$, the root of this reduction is obviously also the root of an black (\resp white) \textsf{Loop} in $G$. Now, suppose that there is at least one contact vertex in $ground(\ol{R}_2)$. Let $w$ the vertex of degree three in $ground(\ol{R}_2)$, and $r,r'$ its two neighbors. Let us focus on the first two contact vertices $c$ and $c'$ met when we sweep the loop from $w$ in each direction (or just $c$ its the only contact vertex present). We claim that at least one of $r$ of $r'$ a the black (\resp white) root of an (alternating) black (\resp white) \textsf{Even-backbone} in $G$. See Figure \ref{fig:helploop} (a) and (b) (\resp (c)). 
Form left to right this vertex is respectively $r'$, $r$ and $r$. 


\begin{figure}
    \centering
    \includegraphics[width=14cm]{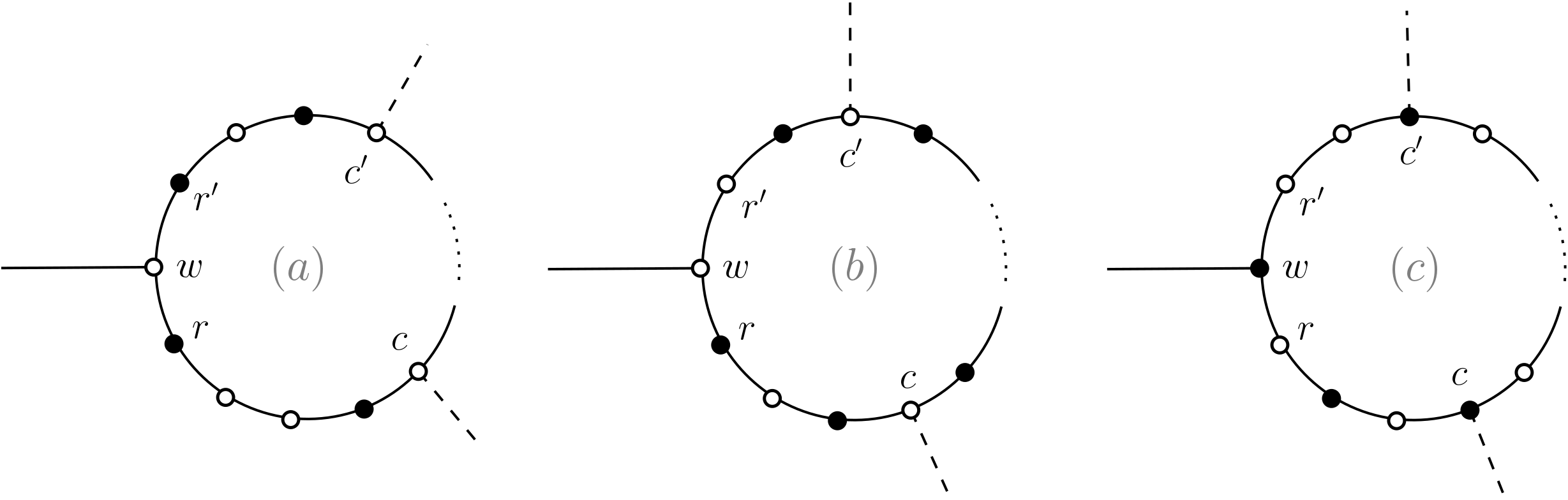}
    \caption{Different types of black, $(a)$ and $(b)$, and white \textsf{Loop} $(c)$. Dashed edges are contact edges. }
    \label{fig:helploop}
\end{figure}

We now turn our attention to the case when $H$ is a bad\footnote{meaning here that $I$ is maximum in $H$} edge or cycle. 
If $H$ is a bad edge, then its black (\resp white) vertex has degree one in $G$, and therefore is the root of black (\resp white) path or branching reduction in $G$. Similarly as the case when $\ol{R}_2$ is a \textsf{Loop}, if $H$ is a bad cycle, 
then there exists a black (\resp white) vertex  $r'\in V(H)$ that is the root of a black (\resp white) \textsf{Loop} in $G$, when $\ol{R}_1$ has one contact vertex in $H$, or the root of a black (\resp white) \textsf{Even-backbone} in $G$, when $\ol{R}_1$ has more that one contact vertex in $H$. 

\end{proof}


We now turn our attention to the case when $\overline{R}_1$ is a backbone reduction.

\paragraph{$\overline{R}_1$ is an even-backbone reduction.}
If the first reduction executed in $G$ is an \textsf{Even-backbone}, it means according to Greedy order, that the graph $G$ does not contain any degree one vertices or any loop reductions. All degree two vertices are contained in some backbones linking two distinct degree three vertices. If the end-points of the backbone of $\overline{R}_1$ are adjacent, then $\overline{R}_1$ satisfies Observation \ref{obs:contact-potential}, so that this case was treated in the previous section. From now, let us assume that these end-points are independent. In the following, we use the same terminology as in Lemma \ref{lem:rulefor26}.

\begin{itemize}
    \item If all contact vertices of $\overline{R}_1$ are white (\resp black), then at most one connected component created by $\overline{R}_1$ has potential $-1$. Indeed, suppose for a contradiction, that $H_i$, with $i\ge 2$, has potential $-1$. By induction hypothesis, it must satisfy assumptions (\ref{prob25}) or (\ref{prob26}) of Lemma \ref{lemma:induction}. According to Claim \ref{claim:induction}, there was a black (\resp white) reduction in $H_i$ before $\overline{R}_1$ was executed. This reduction is neither a degree one nor a loop reduction, so it must be an even-backbone reduction. Since they are black (\resp white), the root vertices of this reduction are distinct than $\overline{R}_1$'s contact vertices
    , which contradicts Lemma \ref{lem:rulefor26}. We proved (\ref{-1}) when all contact vertices of $\overline{R}_1$ are all white (\resp black).
    
    For (\ref{prob}), if $\Phi_I(\calE)=-1$, then one connected component created, say $H_1$ has potential $-1$, which by induction satisfies (\ref{prob25}) or (\ref{prob26}) and $\Phi_I(\overline{R}_1)=0$, so that $\overline{R}_1$ is a white reduction with only white contact vertices (Claim \ref{techclaim:eb}). 
    Claim \ref{claim:induction} guarantees the existence of a black reduction in $G$, so that $G$ is potentially problematic. 
    \item If some of $\overline{R}_1$ contact vertices are are both black and white, then we argue that $\Phi_I(\calE)\ge 0$. First, the potential of $\overline{R}_1$ is at least two\footnote{We assume here that the two end-points of the corresponding backbone are not adjacent.} (Claim \ref{techclaim:eb}). Therefore we should argue that there are \emph{at most two} connected components with potential $-1$. This is true since there are at most two connected components with strictly more than one contact vertex, and at most one connected component with exactly one contact vertex can have potential $-1$. Indeed, for connected components with only one contact vertex, Claim \ref{claim:induction} applies so that we can use the same argumentation than in the previous paragraph.
    
\end{itemize}

\paragraph{$\overline{R}_1$ is an odd-backbone reduction.}
\begin{itemize}
    \item Assume first that $\overline{R}_1$ has potential $\Phi_I(\overline{R}_1)=-1$ (\resp $\Phi_I(\overline{R}_1)=0$). Then, Claim \ref{techclaim:od} indicates that all its contact vertices are white (\resp black). 
    \begin{enumerate}[(1)]
    \item  We prove that all connected components created by $\overline{R}_1$ have potential at least zero. Assume that it is not true for the component $H_1$. By induction it satisfies (\ref{prob25}) or (\ref{prob26}). According to Claim \ref{claim:induction}, there exists a black (\resp white) reduction in $H_1$ before $\overline{R}_1$ is executed which contradicts Greedy order\footnote{This is where the \textsf{Odd-backbone}-rule is used : in any execution we can not have two consecutive odd-backbone reductions with minimum potential. 
    }
    , since any black or white reduction has a priority strictly higher than \textsf{Odd-backbone}.
    \item  When $\overline{R}_1$ is supposed to be bad odd-backbone by assumption, (\ref{prob25}) is always true, and otherwise, if $\Phi_I(\overline{R}_1)=0$, we just proved that $\Phi_I(\calE)\ge 0$.

    \end{enumerate}
    \item Suppose that $\Phi_I(\overline{R}_1)\ge 1$. We claim that at most one component created by $\overline{R}_1$ can have potential $-1$. First note that $\overline{R}_1$ has three contact edges and thus at most three contact vertices.  Indeed, at most one connected component created has at least two contact vertices, and any connected component $H$ created with exactly one contact vertex can not have potential $-1$ since Claim \ref{claim:induction} would imply  the existence of an highest priority reduction in $H$.
   %
   
   
\end{itemize}


This concludes the proof of Lemma \ref{lemma:induction}.%
\end{proof}

\begin{proof}(of Theorem \ref{theo:ratio})
We first treat the case where the input graph has at least one vertex with degree at most $2$. Let $G$ be a connected graph, $I$ a \emph{maximum} independent set in $G$, and $\calE$ an execution. Our goal is to show that the \emph{exact} potential is non-negative :
\begin{equation}\label{eq:goal}
\Phi_{G,I}(\calE)\ge 0
\end{equation}
Suppose this is true. Then from Proposition \ref{prop:sumpotential} we have that $5|\calE|-4|I|\ge 0$, which can be re-written as $\dfrac{|I|}{|\calE|}\le \dfrac{5}{4}$, and since this is true for any independent set, then we have established the desired approximation. 

In Lemma \ref{lemma:induction}, we proved that $\Phi_I(\calE)\ge -1$. Suppose that the inequality is tight. Then, the first reduction is a bad odd-backbone reduction or $G$ is potentially problematic. In any case there exists in $G$ a white vertex with degree at most two --- that is the root of the first \textsf{Odd-backbone} or of the white reduction in $G$ --- so that, using Claim \ref{claim:comparepot} we have:
$$
\Phi_{G,I}(\calE)=\Phi_I(\calE)+\sum_{v \notin I}(3-d_G(v))\ge -1 + 1 =0
$$
Suppose now that all vertices in $G$ have degree three, and assume that $I$ is maximal, so that for any vertex $u$, $|N_G[u]\cap I|\ge 1$. Then, let us consider any degree $3$ vertex $u$ and we see that one of the four executions of the algorithm will be execution with $v \in N_G[u]$ being black, and let us call this first reduction $R^*$.
We detect which of those four executions to take by taking the one that leads to the largest size solution. By Claim \ref{claim:larger-larger} it implies that this execution has the largest potential.

On the other hand let us consider the execution $\calE$ of $R^*$ with its black root $v$.  Let $H_1,\dots, H_s$ denote the connected components created by $R^*$, and $\calE_1,\dots, \calE_s$ the corresponding executions. Without loss of generality we have $\calE=(R^*, \calE_1,\dots,\calE_s)$. To prove that the exact potential of $\calE$ in $G$ is positive, the trick is to consider that the first reduction $R^*$ does not use any loan from its loan edges, so that its potential is exactly 1. This implies also that each connected component will not have any debt edges. Then as we proved before, since each connected component $H_i$ has a vertex with degree at most two, its exact potential in $H_i$ is non-negative. Therefore, we have
$$
\Phi_{G,I}(\calE)=1+\sum_{i=1}^s \Phi_{H_i,I}(\calE_i)\ge 1.
$$
%
%
%
\end{proof}

\noindent
{\bf Remark:} Notice that our analysis implies that the size of the final solution is in fact at least $(4/5)OPT + 1/5$ if the first executed reduction is of degree $3$. 
Moreover, if the first reduction $\ol{R}$ is a bad even-backbone reduction, and $G$ is not a problematic graph, then our analysis proves that the \Greedystar solution's size is a at least $(4/5) OPT + 1/5$.
This precisely matches the size of the lower bound example of Halld{\'o}rsson. Our analysis even indicates that this lower bound example has the worst possible approximation ratio for any ultimate greedy algorithm. 
Indeed, this counter example is actually a sequence of graph $(G_n)_n$, and the solution returned by any Greedy has size $(4/5)OPT(G_n)+1/5$, and therefore the corresponding approximation ratio goes to $5/4$ when the size of $G_n$ goes to infinity. However, one may wonder if there exists a (finite) graph $G$, such that any greedy produces a solution of size at most $(4/5)OPT(G)$. Our previous analysis indicates that such a graph can not exists. 

Indeed,
this graph must satisfy (\ref{prob25}) of (\ref{prob26}) from Lemma \ref{lemma:induction}, otherwise our \Greedystar outputs a solution of size at least $(4/5)OPT(G)+1/5$. Then, for any maximum independent set, this graph must have at least one \emph{black} minimum degree vertex, and in this situation, we could for instance try all possible minimum degree vertices (only for the first step), and pick the execution of maximum size. This modified greedy algorithm returns a solution at least
$(4/5)OPT(G)+1/5$, for any input graph $G$. \\


\noindent
{\bf Remark:} Through basically the same technique, we can achieve a greedy algorithm with an approximation ratio of $\frac{4}{3}$ but with linear running time. Observe that if we set $\gamma = 4$ and $\sigma =3$, the minimum potential of an odd-backbone reduction changes from $-1$ to $0$. Now, only \textsf{Cycle} has potential value of $-1$. With this observation, we are able to extend \textsf{Loop} in a particular way, which implies, finally, that we are able to exclude even-backbone reductions from the definition of black and white type reductions, and also preserve the induction hypothesis in our inductive proof. It implies that the \textsf{Even-backbone} rule is not necessarily needed, and then the running time of such greedy algorithm will be linear.


\subsection{Technical claims}\label{section:technical}

\begin{claim}\label{claim:base-case}
    Given a connected graph $G$, where any execution consists of \emph{one} extended reduction $\ol{R}$, \ie \textsf{Point}, \textsf{Edge} or \textsf{Cycle}, then and if $\Phi_I(\ol{R}) = -1$, for an given independent set $I$ in $G$, then $G$ is either a bad odd-length cycle, or a bad \textsf{Edge}, and $I$ is \emph{maximum} in $G$.
\end{claim}
\begin{proof} 
  First, by Claim \ref{claim:potentialextended}, if $\ol{R}$ is a \textsf{Point} then $\Phi_I(G) \geq 1$. Then, if it is an edge reduction, since it is a basic reduction, it is easy to check that $\Phi_I(\ol{R})=-1$ only when one vertex is black \ie $I$ is maximum, see Figure \ref{fig:negativepot} $(b)$.
  
  Finally, when $G$ is a cycle then, the corresponding cycle reduction has potential $-1$ when all inequalities from equation (\ref{eq:cycle}) in the proof of Claim \ref{claim:potentialextended} are tight. In particular the number $b$ of black vertices must be maximum, and $\left\lfloor\frac{n}{2}\right\rfloor-\left\lceil\frac{n}{2}\right\rceil=1$, that only arises when the length $n$ of the cycle is odd.
\end{proof}

\begin{claim}\label{techclaim:easyred}
  Any \textsf{Path}, \textsf{Branching}, \textsf{Loop} with minimum potential have type black or white, and have contact vertices all black, or all white. 
\end{claim}

\begin{proof}
For (basic) path and branching reductions, worst case potential (Figure \ref{fig:01pot} (a),(h) --- \resp (b),(g)) are white (\resp black) reductions, with white (\resp black) contact vertices. 

Next, as noticed before, the potential of a \textsf{Loop}, is correlated to the potential of the cycle reduction obtained by removing its contact edge, because adding an edge to a ground of reduction either increases by one the loan or decrease by one the debt, so that the potential increases by one or two.
In particular, when \textsf{Loop} has minimum  potential, the corresponding cycle has also minimum potential. By Claim \ref{claim:base-case}, we know that the independent set must be maximum, and its backbone must have odd-length, so that the loop reduction must have type black or white. Moreover when the potential is minimum, the contact can not be incident to two white vertices, otherwise making the contact vertex black would decrease the potential. Therefore, when $\Phi_I(\ol{R})=-1$, the contact vertex has the same type than the reduction.
\end{proof}

\begin{figure}
    \centering
    \includegraphics[width=15cm]{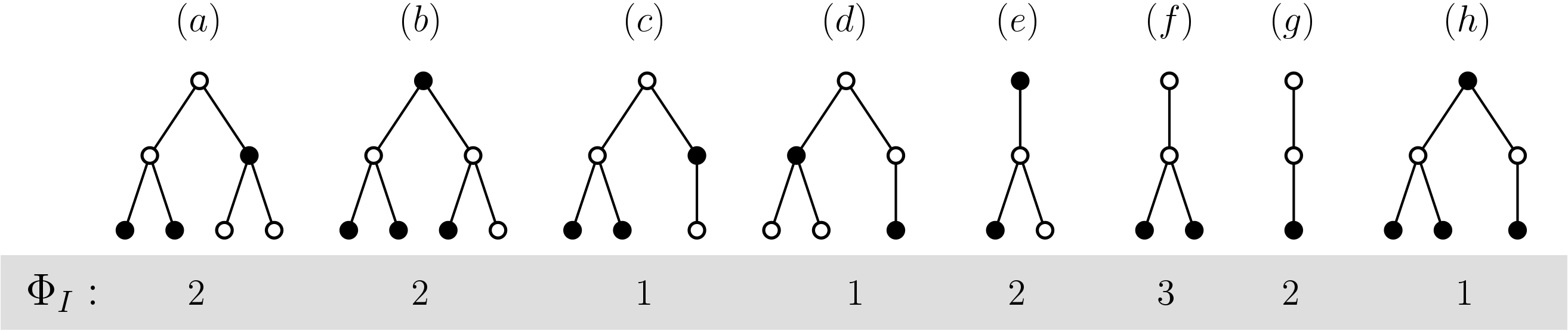}
    \caption{Some examples of good case potential value}
    \label{fig:potothers}
\end{figure}

\begin{claim}\label{techclaim:eb}
Let $\ol{R}$ be \textsf{Even-backbone}, for all independent set $I$ we have $\Phi_I(\ol{R}) \geq 0$, and moreover, 

\begin{enumerate}
    \item if $\Phi_I(\ol{R}) = 0$, then $\ol{R}$ is a white reduction 
    with only white contact vertices.
    \item if the end-points of $\ol{R}$'s backbone are not adjacent, and if its contact vertices are not all white or all black, then $\Phi_I(\ol{R}) \geq 2$.
\end{enumerate}
\end{claim}

\begin{proof}
\begin{enumerate}
    \item Suppose that the potential of an \textsf{Even-backbone} $\ol{R}=(R_1,\dots,R_t)$ is minimum : $\Phi_I(\ol{R})=0$. We know that its end-points are not adjacent, otherwise it has potential at least one, and the potential of all basic reductions $R_i$ must be minimum. When $\ol{R}$ has size one, the worst case arises only when $I$ is such as in Figure \ref{fig:01pot} (e), that is a particular case of white reduction with white contact vertices. For greater sizes, this implies that the first basic reduction $R_1$ is like Figure \ref{fig:negativepot} (d). Then $R_2$ must be a \textsf{Path} with minimum potential and a white root, \ie Figure \ref{fig:01pot} (a), so on and so long for all path reductions. The final branching reduction $R_t$ has minimum potential and a white root, as in Figure \ref{fig:01pot} (h). Finally we proved that this backbone is alternating and the root is white, so that $\ol{R}$ is a white reduction. Moreover, all its end-points must be white.
    \item In the following we will say that a reduction is \emph{mixed} is it has two different type contact vertices. 
    When $\ol{R}$ has size one and is mixed, we can easily check by hand, that its potential is at least two (Figure \ref{fig:potothers} (a), (b)). Consider now an mixed even-backbone  $\ol{R}=(R_1,\dots,R_t)$, and without loss of generality\footnote{recall that we are free to choose the root of \textsf{Even-backbone}}, assume that the last branching reduction $R_t$ has at least one black contact vertex.

    First, if $R_1$ or $R_t$ are mixed, then their potential is respectively at least 
    $1$ and $2$ (see Figure \ref{fig:potothers} (c),(d) and (e)) so that $\Phi_I(\ol{R})\ge 2$. 
    
    Otherwise, assume that $R_1$ and $R_t$ have only respectively white and black contact vertices. In particular the root $r$ of the first \textsf{Path} $R_2$ is white. If the root $r'$ of the last branching reduction is white then\footnote{we assume here that at least one contact vertex is black}, its potential is at least $3$ (Figure \ref{fig:potothers} (f)), so that $\Phi_I(\ol{R})\ge 2$. Then, if $r'$ is black, since the distance between $r$ and $r'$ is even, we must found a \textsf{Path} $R_i$ with both root and middle white vertices. Such a reduction has potential at least $2$ (Figure \ref{fig:potothers} (g)) so that $\Phi_I(\ol{R})\ge 2$.
    
\end{enumerate}
\end{proof}

\begin{claim}\label{techclaim:od}
Let $\ol{R}$ be odd-backbone, for all independent set $I$ we have $\Phi_I(\ol{R}) \geq -1$, and moreover, 

\begin{enumerate}
    \item if $\Phi_I(\ol{R}) = -1$, then it has an alternating backbone, a white root and only white contact vertices.
    \item if $\Phi_I(\ol{R}) = 0$, then it has an alternating backbone, a black root and only black contact vertices.
\end{enumerate}
\end{claim}

\begin{proof}
These assumptions can be easily checked by hand for odd-backbone reduction $\ol{R}$ of size one, see Figure \ref{fig:negativepot} (d) and Figure \ref{fig:potothers} (h). Suppose now that $\ol{R}=(R_1,\dots,R_t)$ has size at least two, where $R_1$ is the basic odd-backbone reduction, and $R_i$ are path reductions, for $i\ge 2$.
\begin{enumerate}
    \item If $\Phi_I(\ol{R})=-1$, then necessarily, $\Phi_I(R_1)=-1$ and for all \textsf{Path}, $\Phi_I(R_i)=0$. The first reduction has only white contact vertices (Figure \ref{fig:negativepot} (d)), and then $R_2$ has a white root, so it must be as in Figure \ref{fig:01pot} (a), in particular it must has a white contact vertex, so that $R_3$ has a white root, and so on and so forth. This implies that the corresponding backbone is alternating, the root vertex and all contact vertices are white.
    \item If $\Phi_I(\ol{R})=0$, then either \textbf{case 1 :} all basic reductions have potential zero, or \textbf{case 2 :} $\Phi_I(R_1)=-1$ and there is one \textsf{Path} $R_j$ with potential one. 
    \begin{description}
    \item[case 1 :] $R_1$ must be like in Figure \ref{fig:potothers} (h), and using exactly the same argumentation than in the previous paragraph, we show that all following path reductions are like in Figure \ref{fig:01pot} (b), so that $\ol{R}$'s backbone is alternating, $\ol{R}$ has black root and contact vertices.
    \item[case 2 :] we show that this case is impossible. Indeed, as before all path reductions $R_i$, with $i<j$ have a white contact vertex (Figure \ref{fig:01pot} (a)), so that $R_j$ has a white root. Since it has potential one, its middle vertex must be white (otherwise its potential is zero, see Figure \ref{fig:01pot} (a)), and in this case its potential is at least two (Figure \ref{fig:potothers} (g)).
    \end{description}
\end{enumerate}
\end{proof}

\begin{claim}\label{claim:larger-larger}
  Let $G$ be a connected graph and $I$ an independent set n $G$. Let $\calE=(\overline{R}_1,\overline{R}_2,\dots,\overline{R}_\kappa)$ be an execution of greedy on $G$, and let this execution produce a solution of size $s_{r}$. Let $\calE'=(\overline{R}'_1,\overline{R}'_2, \dots,\overline{R}'_{\kappa'})$ correspond to another execution of greedy on $G$ producing a solution of size $s_{r'}$. Then if $s_r > s_{r'}$ then $\Phi_I(\calE) > \Phi_I(\calE') $, and otherwise, if $s_r < s_{r'}$ then $\Phi_I(\calE') > \Phi_I(\calE)$.
\end{claim}
\begin{proof}
  We know that $\Phi_I(\calE) = 5 s_r - 4 | I \cap V(G) | $, and $\Phi_I(\calE') = 5 s_{r'} - 4 | I \cap V(G) |$. Therefore, $\Phi_I(\calE) < \Phi_I(\calE')$ if and only if $s_r < s_{r'}$. 
\end{proof}

\section{Minimum vertex cover in subcubic graphs}

In this section we will apply our greedy algorithm \Greedystar to obtain a fast $6/5$-approximation algorithm for the minimum vertex cover (MVC) on subcubic graphs. 

\begin{algorithm}[H]
\SetAlgoVlined
\setcounter{AlgoLine}{0}

 \KwIn{A subcubic graph $G$.}
 \KwOut{A vertex cover $C$ of $G$.}
 
  Run algorithm \Greedystar on $G$, and let $J$ be the output independent set.
  
  Let $C \leftarrow V(G) \setminus J$.
  
  \Return C.
  \caption{ Complementary greedy algorithm for MVC.}
 \label{algo:coml-greedy}
\end{algorithm}

We will first show that a direct application of \Greedystar leads to a $5/4$-approximate vertex cover. We will prove this by an interesting ``dual'' interpretation of our ``primal'' potential used in the analysis of \Greedystar.

\begin{theo}\label{theo:vertex-cover-1.25}
The complementary greedy algorithm of Algorithm \ref{algo:coml-greedy} achieves a $5/4$-approximation ratio for the minimum vertex cover problem on subcubic graphs, and has running time $O(n^2)$, where $n$ is the number of vertices in the graph.
\end{theo}


Let $G$ be a sub-cubic graph, $I$ a maximal independent set in $G$ and $C = V(G) \setminus I$ be a vertex cover in $G$.
  For the vertex 
cover problem, we will define a potential function $\Psi_C(R)$ of reduction $R$. This reduction is the same than in MIS but now it is executed as for the MVC problem.

For instance, when $R$ is basic reduction (b) in Figure \ref{fig:potothers} then its black root vertex is in $I$ in case of MIS. But in the case of MVC, the two white neighbors of this black root are in $C$. 
A vertex is white, if it is in $C$, and a vertex is black, if it is not in $C$. Then, we define the \emph{loan} edge $e$ of a reduction $R$ as a contact edge with a white contact vertex. We also define the \emph{debt} of a \emph{white} vertex in the ground of $R$ as the number of times this vertex was incident to a loan edge, let us call it $e'$, of a reduction that was previously executed. Such loan edge $e'$ is also called a \emph{debt} edge of reduction $R$.%

In analogy to MIS problem, we define the exact potential of a reduction $R$ for the MVC problem as:
\begin{equation*}
\Psi_{G,C}(R):= 4 \cdot (|ground(R)| -1) - 5\cdot|C \cap ground(R)|-
loan_{C}(R) + debt_{G,C}(R).
\end{equation*}
Similarly, the potential of reduction $R$ for MVC problem is defined as:
\begin{equation*}
    \Psi_{C}(R):=4\cdot (|ground(R)| -1) -5\cdot|C\cap ground(R)|-loan_{C}(R)+ debt_{C}(R),
\end{equation*}
and observe that $\Psi_{G,C}(R) \leq \Psi_{C}(R)$.
\ignore{
We present definition of \emph{loan} and \emph{debt} edges for minimum vertex cover version. The loan edge $e$ of a reduction $R$ is the edge such that $e \in E_\emph{contact}$ and $f(N(e) \cap V_R) = 0$; and a \emph{debt} edge of $R$ is the edge such that $e \in E_\emph{past}$ and $f(N_G(e) \setminus V_R) = 0$. The situation is same as in MIS problem that for sequence of reduction $S$ executed by complementary algorithm algorithm on graph $G$, we have: $\sum_{R \in S}\sum_{e \in E_\emph{loan}(R)} 1 = \sum_{R \in S}\sum_{e \in E_\emph{debt}(R)} 1$ or $\sum_{R \in S} E_\emph{loan}(R) = \sum_{R \in S} E_\emph{debt}(R)$.
Then, we define potential function $\Psi(R)$:
\begin{equation}
 \Psi(R) = 4|SOL(R)| - 5|OPT(R)| - |E_\emph{loan}(R)| + |E_\emph{debt}(R)| 
\end{equation}
}
\ignore{
\begin{align}
    \Psi(S) &= \sum_{R \in S} \Psi(R) = \sum_{R \in S} (4|SOL(R)| - 5|OPT(R)|) - |E_\emph{loan}(S)| + |E_\emph{debt}(S)|  \\
    &= \sum_{R \in S} (4|SOL(R)| - 5|OPT(R)|)
\end{align}
Therefore, if we can prove for any $S$ executed from a graph $G$ such that $\Phi(S) \leq 0$,
\begin{equation*}
    \sum_{R \in S} (4|SOL(R)| - 5|OPT(R)|) \leq 0
\end{equation*}
\begin{equation*}
      |SOL(S)| \leq \frac{5}{4} |OPT(S)|
\end{equation*}
This proves the $\frac{5}{4}$-approximation ratio.
}%
To finalise our proof, we prove the following ``duality'' lemma.
\begin{lem}[Duality lemma]\label{l:duality-lemma}
Given a sub-cubic graph $G$, an independent set $I$ in $G$, and $C=V(G)\setminus I$ a vertex cover,  for any basic reduction $R$ distinct than an isolated vertex reduction, we have $-\Phi_I(R) \geq \Psi_C(R)$ and $-\Phi_{G,I}(R) \geq \Psi_{G,C}(R)$.
\end{lem}
\begin{proof}
Recall the definition of the potentials:
$$
\Phi_{I}(R):=5-4\cdot|I\cap ground(R)|+
loan_{I}(R)-debt_I(R),
$$ and 
$$
\Psi_{C}(R):=4\cdot (|ground(R)| -1)-5\cdot|C\cap ground(R)|-loan_{C}(R)+debt_{C}(R).
$$ For any reduction $R$, it is easy to check the following statement by the definition of loan and debt edges: $\emph{loan}_I(R) = \emph{loan}_C(R)$ and $\emph{debt}_I(R) = \emph{debt}_C(R)$. Then, we have: $$- \emph{loan}_C(R) + \emph{debt}_C(R) =  - ( \emph{loan}_I(R) - \emph{debt}_I(R)),
$$ 
Therefore, if we can show for any reduction $R$ we have: 
\begin{equation*}
    5 - 4 \cdot |I\cap ground(R)| \leq  -(4 \cdot (|ground(R)-1) - 5\cdot | C \cap ground(R)|),
\end{equation*}
then the lemma is proved. Observe that we have the following: $I\cap ground(R) = ground(R) \setminus (C\cap ground(R))$. By this identity, the last inequality is equivalent to:
\begin{equation}\label{formula-daul}
    1 + | I \cap ground(R)| \leq |ground(R)|.
\end{equation}
Observe that $|ground(R)| \leq 3$ and $| I \cap ground(R)| \leq 2$ for any reduction $R$. Now we have the following cases and in each case we see that inequality (\ref{formula-daul}) holds:
\begin{enumerate}
     \item For a single edge reduction, $| ground(R)| =2$, and $1+ |  I \cap ground(R)| \leq 2$.
    \item For a triangle reduction, $| ground(R)|= 3$, and $1+ |  I \cap ground(R)| \leq 2$.
    \item For an even- or odd-backbone reduction, $| ground(R)| =3$, and $1+ |  I \cap ground(R)| \leq 3$.%
\end{enumerate}
Notice that the argument for the inequality $-\Phi_{G,I}(R) \geq \Psi_{G,C}(R)$ among exact potentials is basically the same. 
\end{proof}
To complete the proof, we need to show that the isolated vertex reduction does not affect our argument. 
\begin{claim}
  For an isolated vertex reduction $R$, $\Psi_C(R) \leq 0$.
\end{claim}
\begin{proof}
If $R = \{v\}$ is white, then $\Psi_C(R) \leq -2$. And if $v$ is black, then $\Psi_C(R) = 0$.
\end{proof}
\ignore{
  For branching reductions. There are two classification. For those good reduction $R$ in sense of independent set (abbr. good in $I$):
  $\Psi(R) = -2 - e_3 + e_1$, and $e_1 \leq 1$, thus, $\Psi(R) \leq -1$. For those bad reduction $R$ in $I$:
  $\Psi(R) = 3 - e_3 + e_1$, and since it is bad in $I$, then $e_3 \geq 3$, and $e_1 \leq 1$, therefore, only bad-$(2,5)$ reduction $R$ has positive $1$ value for $\Psi(R)$. Notes that $\Phi(R) = -1$, thus, it satisfies.
  
  For cycle reductions, it is easy to show that $\Psi(R) \leq 0$.
  
  For single edge reduction, only isolated single edge reduction $R$ such that $\Psi(R) =1$. Observe that $\Phi(R) = -1$, thus, it satisfies.
  }


\begin{proof} (Theorem \ref{theo:vertex-cover-1.25}) We have shown in the proof of Theorem \ref{theo:ratio} that if $\calE$ is any execution of \Greedystar on a subcubic graph $G$, then $\Phi_{G,I}(\calE) \geq 0$. Thus if we consider the execution $\calE$ as an execution of the complementary greedy algorithm of Algorithm \ref{algo:coml-greedy} then by the Duality Lemma \ref{l:duality-lemma}, we obtain that $\Psi_{G,C}(\calE) \leq 0$ which proves the theorem.
\end{proof}

We will employ in our algorithm a preprocessing technique due to Hochbaum \cite{HOCHBAUM1983243}, which is based on the following theorem.

\begin{theo}[Nemhauser-Trotter \cite{Nemhauser1975}]\label{theorem-N-T}
 For any graph $G = (V,E)$, there is a way to compute the partition $\{V_1,V_2,V_3\}$ of $V$ with time complexity of the bipartite matching problem, such that:
 \begin{enumerate}
     \item There is a maximum independent set $I$ in $G$ containing all of the nodes of $V_1$ but none of $V_2$, \ie $I \cap V_1 = V_1$ and $I \cap V_2 = \emptyset$.
     \item There is no edge between $V_1$ and $V_3$, \ie $N(V_1)\subseteq V_2$.
     \item $\alpha(G(V_3)) \leq \frac{1}{2}|V_3|$
 \end{enumerate}
 Note that if graph $G$ is subcubic, then the running time of computing such partition is $O(n^\frac{3}{2})$.
\end{theo}

\begin{theo}\label{theorem-cg-6/5}
  The complementary greedy algorithm of Algorithm \ref{algo:coml-greedy} combined with the Nembauser-Trotter reduction achieves a $\frac{6}{5}$-approximation ratio with running time $O(n^2)$ for the minimum vertex cover problem on subcubic graphs.
\end{theo}

\begin{proof}
We apply the Nemhauser-Trotter reduction from Theorem \ref{theorem-N-T} to $G =(V,E)$, and $V$ is partitioned into $V_1,V_2,V_3$. Then we run the $\frac{5}{4}$-approximation greedy algorithm \Greedystar on $G_3 = G(V_3)$, and we choose the complement of the independent set $J$ output by the algorithm. Then let $C_3 = V_3 \setminus J$ denote the resulting vertex cover in $G_3$. Let also $C^\ast$ denote a minimum vertex cover of $G$, and $I^\ast$ be a maximum independent set in $G$. 

Observe that for (any) graph $G_3$, we have that $V(G_3) \setminus I^\ast$ is the complement of a maximum independent set on $G_3$, so its size is equal to the size of the minimum size of a vertex cover in $G_3$. Analogously, let $I_3^\ast$ be a maximum independent set in $G_3$ and $C_3^\ast = V_3 \setminus I_3^\ast$ be a minimum vertex cover in $G_3$.

Thus, we have: $|C^\ast| = |V_2| + |C_3^\ast| = |V_2| + (|V_3| - \alpha(G_3))$ by Theorem \ref{theorem-N-T}.
The vertex cover of $G_3$ computed by the algorithm is $C_3 = V_3 \setminus J$. 
By Theorem \ref{theo:ratio}, we have: $|J| \geq \frac{4}{5} \alpha(G_3)$. By Theorem \ref{theo:ratio}, $\alpha(G_3) \leq \frac{|V_3|}{2}$, thus,  $0 \leq \frac{|V_3|}{5} - \frac{2}{5}\alpha(G_3)$, and
\begin{equation*}
    |C_3| \leq |V_3| - \frac{5}{4} \alpha(G_3) \leq |V_3| - \frac{4}{5} \alpha(G_3) + \frac{|V_3|}{5} - \frac{2 \alpha(G_3)}{5} = \frac{6}{5} \cdot (|V_3| - \alpha(G_3)).
\end{equation*}
Our algorithm outputs $C = V_2 \cup C_3$ as the vertex cover in $G$, and we have that
\begin{equation*}
    |C| = |V_2| + |C_3| \leq \frac{6}{5}|V_2| + \frac{6}{5} |C_3^\ast| \leq \frac{6}{5} |C^\ast|,
\end{equation*}
and this concludes the proof of the approximation guarantee. The running time bound follows from Theorems \ref{theo:ratio} and \ref{theorem-N-T}.
\end{proof}

\section{Lower bounds}\label{s:lower_bounds}

Given a graph $G=(V,E)$, we say that $I$ is a \emph{greedy set} of $G$, if $I$ is an independent set, and its elements can be ordered, $I=\{v_1,\dots,v_k\}$ in such a way that, for all $1\le i\le k$, the vertex $v_i$ has minimum degree in the subgraph $G_i$, where $G_i:= G[V\setminus N_G[\{v_1,\dots, v_{i-1}\}]]$. The size a maximum greedy set in $G$ is denoted by $\alpha^+(G)$.

\subsection{Ultimate lower bounds for high degree graphs.}

\begin{theo}\label{theo:deltaapx}
The ratio of any Greedy like algorithm in graphs with degree at most $\Delta$ is at least $\dfrac{\Delta+1}{3}-\calO(1/\Delta)$.
\end{theo}

This result is an extension of Theorem 6 in \cite{Halldorsson1997}. In this result Halld{\'{o}}rsson and Radhakrishnan present examples where the ratio between the worst execution of the basic Greedy and the optimal independent set is $\dfrac{\Delta+2}{3}-\calO(\Delta^2/n)$. However, on these examples there exists several vertices with minimum degree and picking the right minimum degree vertex could lead greedy to an optimal solution. Equivalently, it means that there exists graphs of bounded-degree where the \emph{minimum} greedy set is small compare to the maximum independent set. We prove that we can the observation to the \emph{maximum} greedy set while keeping roughly the same ratio. Our extension of these examples consists in increasing the degree by one of some vertices of this graphs so that any greedy set has the same size and the corresponding ratio is $\dfrac{\Delta+1}{3}-\calO(1/\Delta)$. 

\begin{figure}
    \centering
    \includegraphics[width=0.7\textwidth]{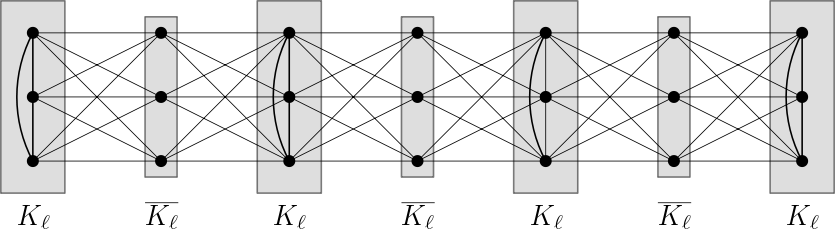}
    \caption{An example when $\ell =3$. $K_\ell$ and $\overline{K_\ell}$ respectively denotes a clique and an independent set of size $\ell$.}
    \label{fig:lowerboundanygreedy}
\end{figure}
\begin{proof}
We show this construction for the case $\Delta \equiv 2 \pmod 3$. See Figure \ref{fig:lowerboundanygreedy}. Let $\ell$ be the integer such that $3\ell-1=\Delta$. The graph consists in a chain of subgraphs, alternating a clique on $\ell$ vertices and an independent set of size $\ell$. Each subgraph is completely connected with the  adjacent subgraphs in the chain. This structure ends with a complete graph on $\ell$.   vertices. 
The degree of the vertices in the extreme clique is $2\ell-1$, while the degree of vertices of other cliques and independent sets are respectively $\Delta=3\ell-1$ and $2\ell$. Any greedy algorithm pick one vertex in each clique while the optimal solution is the union of all vertices in the independent sets. If $n$ denotes the number of vertices in the graph, the ratio between the size of the optimal solution and the size of the solution returned is 
$$
\dfrac{(n-\ell)/2}{(n-\ell)/2\ell +1}=\ell - \dfrac{\ell}{(n-\ell)/2\ell+1}= \ell - \calO(\ell^2/n)= \dfrac{\Delta+1}{3}-\calO(\Delta^2/n)
$$
In particular for any instance where $n = \Omega(\Delta^3)$, we get the expected result.

\begin{figure}
    \centering
    \includegraphics[width=\textwidth]{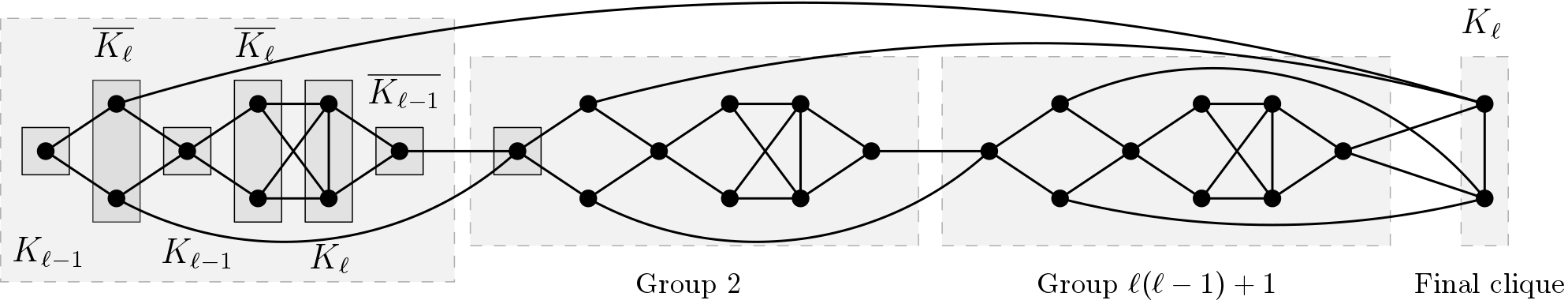}
    \caption{The construction when $\Delta=3\ell-2$. }
    \label{fig:delta1}
\end{figure}

For the case $\Delta \equiv 1 \pmod 3 $, we need a more complicated graph that can be describe as a chain of groups of six subgraphs. Consider the integer $\ell$ such that $3\ell -2 = \Delta$. Each group is formed by a chain of subgraphs of size $\ell$ or $\ell-1$ that are alternatively a clique and an independent set. The total graph consists in a chain of $\ell(\ell-1)+1$ such groups where the last independent set and the first clique of the next group are completely connected. 
Then, this chain ends with a clique of size $\ell$ totally connected with the last independent set of the last group. Additionally, add a matching of size $\ell-1$ between the first independent set of each group to the first clique of the next group. Since these independent sets have size $\ell$, there is one unmatched vertex per each such independent set. 
Finally add an edge from each of these vertices to the final clique. It is not difficult to see that this can be done so that all vertices of the final clique have degree $3\ell-2=\Delta$. See Figure \ref{fig:delta1}. We can see that on this graph, the maximum degree is $\Delta=3\ell-2$, the vertices of the first clique of the first group have degree $2\ell-2$, while all independent set vertices have degree $2\ell -1$. It is not difficult to check that any greedy algorithm will pick one vertex in each clique for a total of $3(\ell(\ell-1)+1)+1$ vertices while the maximum independent set consists of the union of all independent sets from each group. This number is $(3\ell-1)(\ell(\ell-1)+1)$. The corresponding ratio is 
$$
\dfrac{3\ell-1}{3}-\dfrac{3\ell-1}{9(\ell(\ell-1)+1/3)}=\dfrac{\Delta+1}{3}-\calO(1/\Delta)
$$

The case $\Delta \equiv 0 \pmod 3$ is treated similarly than the previous one, using instead the following group
$$
K_{\ell-1} - \overline{K_{\ell+1}} - K_{\ell} - \overline{K_{\ell}} - K_{\ell} - \overline{K_{\ell}} 
$$
and where matchings are between the first independent set and the first clique of the following group and between the last independent set and the last clique of the next group. Details of the construction and calculation are left to the curious reader.
\end{proof}

\subsection{Cubic planar graphs}



The maximisation problem \pr{MaxGreedy} consists in finding a maximum size greedy set in a given graph $G$. This problem was shown to be NP-hard \cite{BODLAENDER1997101}. We first prove that this problem remains NP-hard in the very restricted class of planar cubic graphs.




\begin{theo}\label{hardcubic}
\pr{MaxGreedy} is NP-complete for planar cubic graphs.
\end{theo}

The proof is a reduction from \pr{MIS} in cubic planar graphs, which is NP-hard \cite{garey}.

\begin{proof}
 Let $G=(V,E)$ be a cubic planar graph with $m$ edges.  
 Let us construct a graph $G'$ by replacing each edge $uv\in E$ by the structure $\mathcal{H}_{uv}$ described in \hyperref[edge]{Figure \ref{edge}}. 
We call
$$
V':=V(G')=V\cup\bigcup_{e\in E}V(\mathcal{H}_e)
$$
and 
$$
E':=E(G')=\bigcup_{uv\in E}\left(E(\mathcal{H}_e)\cup\{au,gv\}\right)
$$
where $au$ and $gv$ correspond to the edges connecting $u$ and $v$ to the graph $\mathcal{H}_{uv}$.

\begin{figure}[h]
\vspace{5pt}
 \centering
  \includegraphics[width=8.8cm]{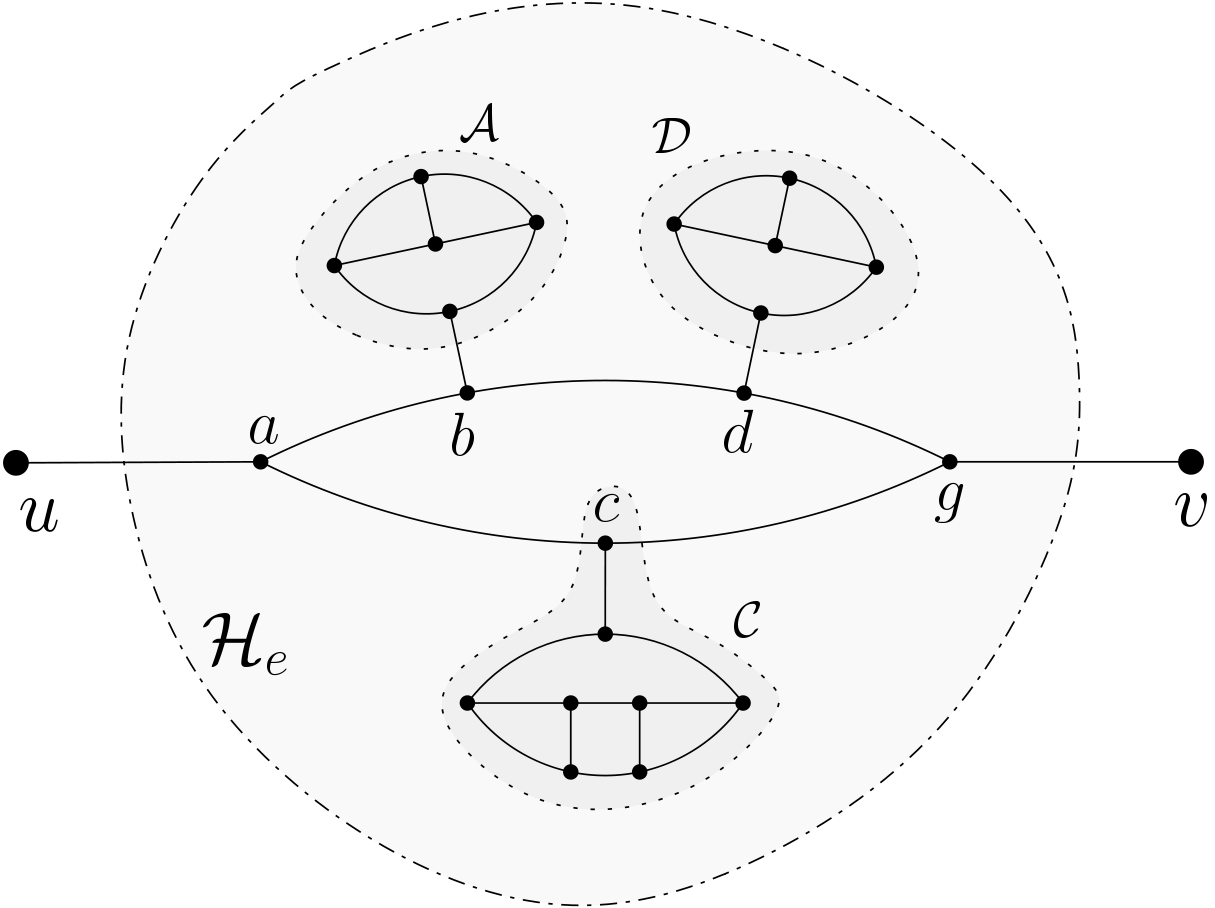}
  \caption{Each edge $e=uv$ is replaced by this gadget $\mathcal{H}_e$}\label{edge}
 \vspace{-0pt}
\end{figure}

$G'$ has order $|V|+22m$ and can be computed in polynomial time.
 

\begin{claim}
 Let $S'\subseteq V'$ be an independent set in $G'$ and $e=uv\in E$. Then, $|V(\mathcal{H}_e)\bigcap S'|\le 9$. Moreover, if both $u$ and $v$  belong to $S'$ then $|V(\mathcal{H}_e)\bigcap S'|\le 8$.
\end{claim}
We can easily check that $|\mathcal{A}\cap S'| \le 2$, $|\mathcal{D}\cap S'| \le 2$, $|\mathcal{C}\cap S'| \le 3$ and $|\{a,b,d,g\}\cap S'| \le 2$. Thus, $|V(\mathcal{H}_e)\bigcap S'|\le 9$. Moreover, if both $u$ and $v$  belong to $S'$ then $\{a,g\}\cap S'=\emptyset$ and then $|\{a,b,d,g\}\cap S'| \le 1$ which gives $|V(\mathcal{H}_e)\bigcap S'|\le 8$.

\begin{claim}
$\alpha(G')\le \alpha(G)+9m$
\end{claim}
Let $S'\subseteq V'$ be an independent set in $G'$. Denote by  $F$ the set of edges of $G$ which have both end nodes in $S'$. We have 
$$
|S'\cap V|-|F|\le \alpha(G)
$$
Indeed, $(S'\cap V)\setminus\{x_e, e\in F\}$ is an independent set in $G$ where $x_e$ is one of the two vertices incident to an edge $e\in F$. Then,  
$$|S'\cap V|-|F|\le|(S'\cap V)\setminus\{x_e, e\in F\}|\le \alpha(G)$$

It follows that
\begin{align*}
 |S'| &= |S'\cap V|+\sum_{e\in E} |\mathcal{H}_e\cap S'| \\
    &\le (\alpha(G)+|F|)+(\sum_{e\in F} |\mathcal{H}_e\cap S'| + \sum_{e\notin F} |\mathcal{H}_e\cap S'| )\\
    &\le (\alpha(G)+|F|)+(\sum_{e\in F} 8 + \sum_{e\notin F} 9 )\\
    &= (\alpha(G)+|F|)+(\sum_{e\in E} 9 - \sum_{e\in F}1 )\\
    &=\alpha(G)+9m
\end{align*}
Since this inequality is true for any independent set $S'$, we have $\alpha(G')\le\alpha(G)+9m$.

\begin{claim}
 There exists a greedy set $S'$ in $G'$ of size $\alpha(G)+9m$.
\end{claim}

Let $S$ be a maximum independent set in $G$. Construct the set $S'$ as follows
\begin{itemize}
 \item While there exists some unpicked nodes in $G'$ do
    \begin{enumerate}
    \item If there exists an unpicked vertex $u\in S$ with minimum degree, add $u$ to $S'$ and nodes $b$ and $g$ in all adjacent gadgets $\mathcal{H}_{uv}$ (see \hyperref[edge]{Figure \ref{edge}})
    \item Otherwise, there exists a vertex of type $a$ with minimum degree in some gadget $\mathcal{H}_{uv}$. 
    \begin{itemize}
      \item If $v\in S$, add $a$ and $d$ to $S'$ 
      \item If $v\notin S$, add $a$ and $g$ to $S'$
    \end{itemize}
    \end{enumerate}
  \item Run \algo{Greedy} on the remaining connected components $\mathcal{A},\mathcal{D}$ and $\mathcal{C}$ of graphs $\mathcal{H}_e$ which have not been picked yet.
\end{itemize}
At the end, we have $S'\cap V=S$ and $|S'\cap V(\mathcal{H}_e)|=9$ for all $e$ in $E$. Then the greedy set $S'$ has the desired size.

Therefore, $\alpha^+(G')=\alpha(G')=\alpha(G)+9m$ and then for any integer $k$ we have
$$
\alpha^+(G')\ge k+9m \text{\hspace{5pt} if and only if \hspace{5pt}} \alpha(G) \ge k.
$$
\end{proof}


\subsection{Hardness of approximation}

For any class of graphs, one can find worst case examples for any greedy algorithm. In the following, we will call such examples \emph{hard graphs}. For such graphs, there is an unique greedy set, meaning that at any stage of the algorithm the choice of the minimum degree vertex is unique, and the ratio of its size to that of the maximum independent set is minimized. For a given class of graphs, we will call this ratio an \emph{ultimate lower bound}, and it shows the limitation of our initial greedy rule. However, since our original motivation was to design additional advices for \Greedy, in order to measure the difficulty of designing such advices, we need to compare ourselves to the maximum greedy set in the given graph. And for these examples, since the choice is unique, \Greedy is optimal when we compare to the maximum size of a greedy set.
In what follows, we show that \pr{MaxGreedy} is hard to approximate in different classes of graphs, within to an inaproximability factor that matches the ultimate lower bound. 
Notice that since the size of a maximum greedy set is upper bounded by the size of a maximum independent set, a lower bound on the approximability of \pr{MaxGreedy} is necessarily smaller than the best approximation ratio achieved by one particular advised greedy algorithm. This suggests that our inapproximability results are (almost) tight. 

A way of proving these inapproximability lower bounds is the following. Given a class of graphs, we first find a hard graph $B$, where the greedy set is unique. In particular, there exists exactly one minimum degree vertex $r$ in $B$, that we call the \emph{root} of the hard graph. We ask additionally these hard graphs to have the following property. If the root is removed from $B$, then there exists an advised greedy algorithm, that proceeds in polynomial time --- for our case, this is simply \Greedy --- that outputs a maximum independent set: $\Greedy(B\setminus r)=\alpha(B)$. 

Now, we add an \emph{anchor graph} $H$ connected to $B$ by the root vertex $r$. Intuitively, the sub-graph encodes an instance of an NP-hard problem, in such way that there exists a greedy set that does not contain $r$ if and only if the instance is positive. 

When the size of $H$ is arbitrarily small compared to $\alpha(B)$, then the gap introduced will be arbitrarily close to the approximation ratio of \Greedy in $B$, \ie the ultimate lower bound for the class of graphs considered.

Precisely, let $\varphi$ be a formula of \pr{SAT}, with $n$ variables and $m$ clauses. Without loss of generality, we can assume that :
\begin{enumerate}[(1)]
    \item Each clause contains two or three literals.
    \item Each positive (\resp negative) literal appears in exactly two (\resp one) clauses.
\end{enumerate}

\begin{figure}
    \centering
    \includegraphics[width=14cm]{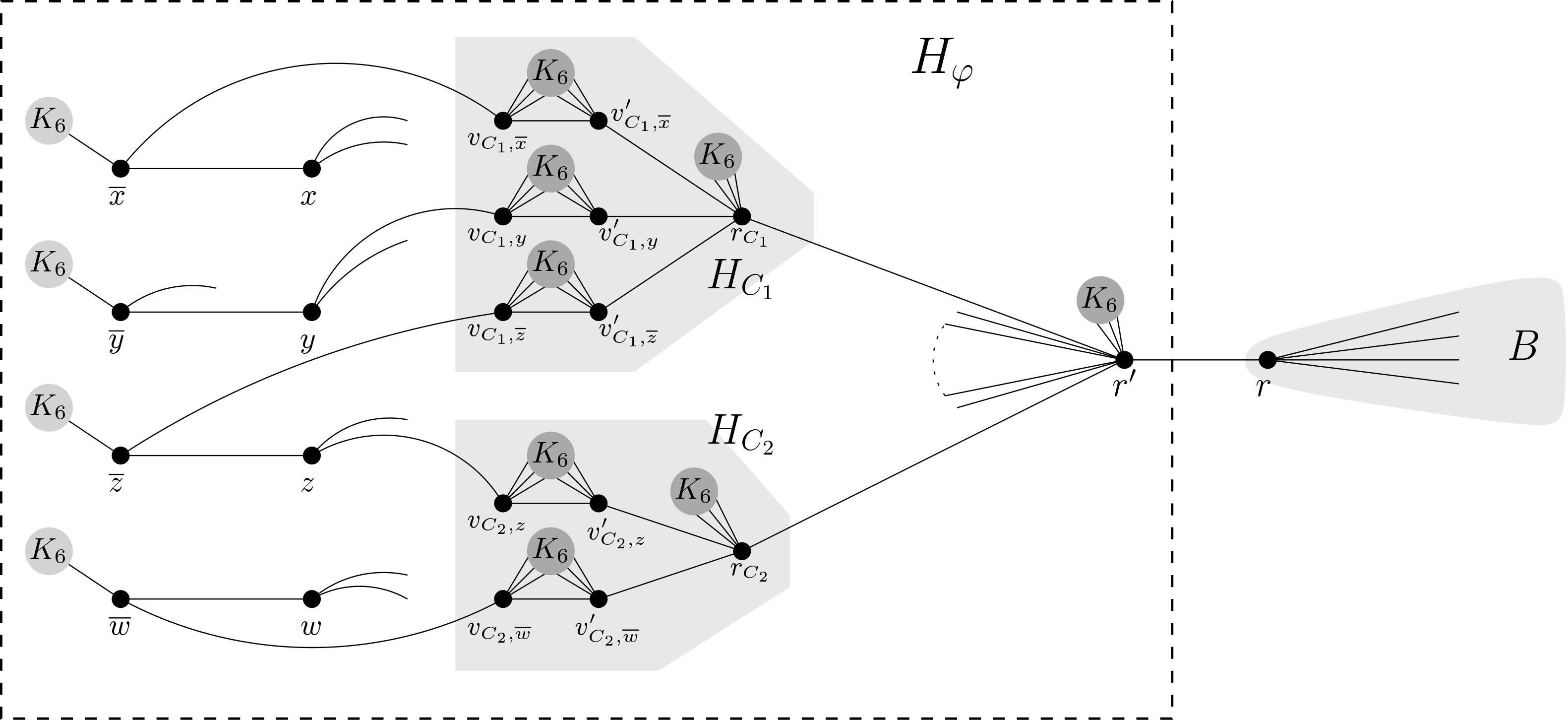}
    \caption{A example of the construction of $H_\varphi$, where $\varphi$ contains clauses $C_1=\ol{x}y\ol{z}$ and $C_2=z\ol{w}$. Cliques $K_6$ ensure that vertices of this graph are executed by \Greedy in the right order.}
    \label{fig:Hphi}
\end{figure}

Indeed, given a \pr{3-SAT} formula, if a variable $y$ occurs $k$ times, then we replace each occurrence by new variables $y_1,\dots, y_k$, and add the new clauses $y_1 \Rightarrow y_2, \dots, y_k \Rightarrow y_1$. Now, (1) is satisfied and each variable appears exactly three times. Finally, if a variable does not satisfy (2), simply exchange its positive and negative literal. This new formula has polynomial size in $n,m$ and is satisfiable if and only if the original formula is satisfiable.

Now we build the anchor graph $H_\varphi$ as follows. First create a vertex $r'$. Create two adjacent \emph{literal vertices} $x$ and $\ol{x}$, for each variable $x$ and create a \emph{gadget} $H_C$ (see Figure \ref{fig:Hphi}) for each clause $C$. In particular, this gadget has one vertex $v_{C,\ell}$ for each literal $\ell$ in $C$ and vertex $r_C$. Connect $v_{C,\ell}$ to the corresponding literal vertex $\ell$ and each $r_C$ to $r$.  Finally, we increase the degree of negative literal vertices by one, so that each literal vertex has now exactly degree three. 
Let $G$ be the graph obtained by connecting $H_\phi$ and $B$ with their respective vertices $r'$ and $r$. See Figure \ref{fig:Hphi}.

Any execution of \Greedy in $G$, resulting in a greedy set $S$ at the end, consists of three phases. 
During phase 1, the minimum degree vertices are the literal vertices, with degree three. \Greedy has to decide either to pick $x$ or $\ol{x}$, for each literal $x$. This is exactly choosing a valuation $\nu : V(\varphi)\rightarrow \{false,true\}^n$ for each variable, such that $x$ is in $S$ if and only if $\nu(x)=true$. 
During the second phase, 
all remaining vertices in $H$ are removed, such that at the end, $r'$ is in $S$ if and only if $\varphi$ is true for $\nu$. To see this, notice that $v_{C,\ell}$ or $v'_{C,\ell}$ have minimum degree four, and the clause $C$ is not satisfied by $\nu$ if and only if for all literals $\ell$ in $C$, the corresponding vertices $v_{C,\ell}$ are still present in the graph. This implies picking all these vertices, and later the vertex $r_C$ in $S$, so that $r'\notin S$.
Finally, the last phase consists in executing the greedy algorithm in $B$: if the formula is not satisfied then, the root $r$ is picked and the number of vertices picked during this phase is $\Greedy(B)$, and otherwise $\alpha(B)$ by the assumption.


To summarize, we built a graph $G$ such that
\begin{itemize}
    \item If $\varphi$ is satisfiable then, there exists a greedy set of size at least $\alpha(B)$.
    \item otherwise, if $\varphi$ is not satisfiable then, any greedy set has size at most $\vert V(H_\varphi)\vert + \Greedy(B)$,
\end{itemize}
where $\vert V(H_\varphi)\vert=\Theta(m+n)$. 
Then, when the hard graph has arbitrarily large size, one can introduce a gap arbitrarily close to the ratio $\alpha(B)/\Greedy(B)$. 
With hard graphs in Figure \ref{fig:hardgeneral} (general graphs), Figure \ref{fig:hardbipartite} (bipartite graphs) and the hard graphs of bounded degree presented in Theorem \ref{theo:deltaapx}, we obtain the following results.

\begin{figure}
 \begin{minipage}[b]{.36\linewidth}
 \centering
  \includegraphics[width=7cm]{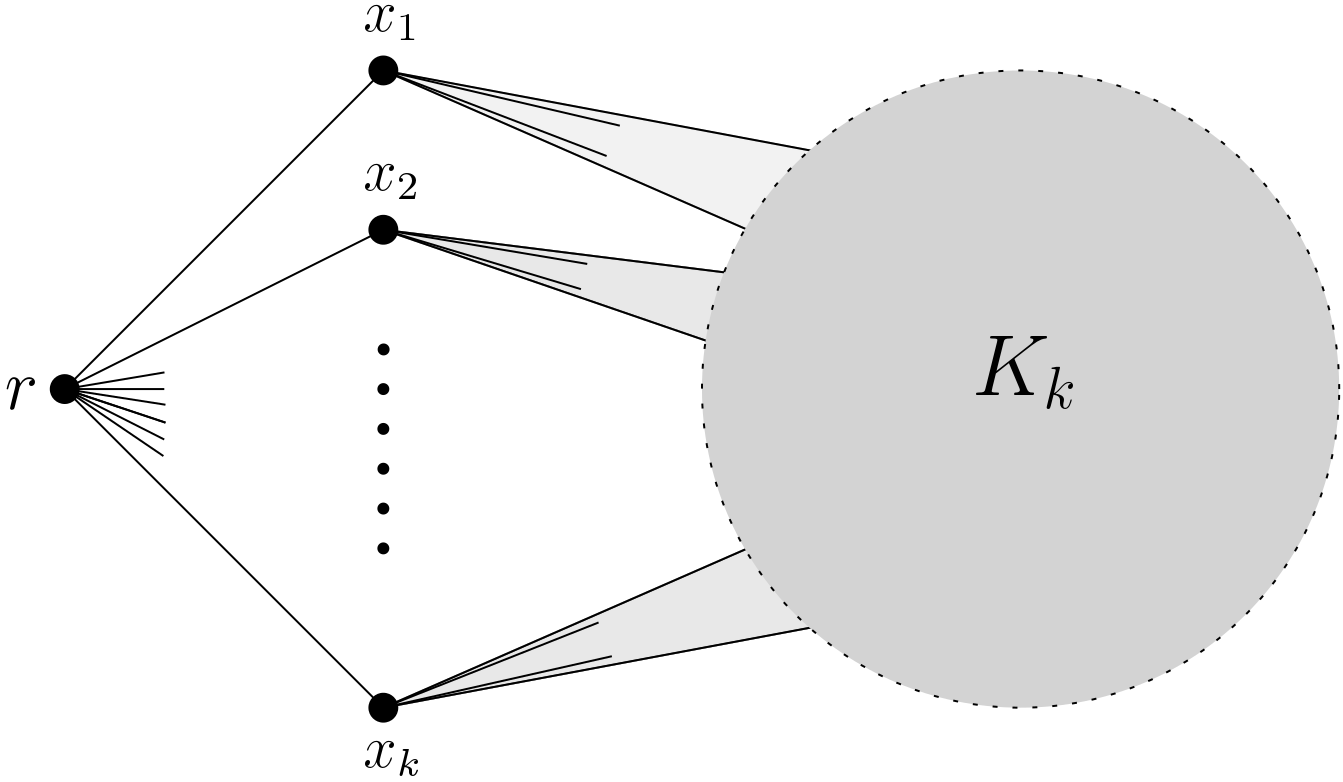}
  \caption{hard general graph. \Greedy picks $r$ and one vertex from the clique, while the maximum independent is $\{x_1,\dots,x_k\}$. The ratio is $\frac{k}{2}= \frac{n-1}{4}$ \label{fig:hardgeneral}}
 \end{minipage} \hfill
 \begin{minipage}[b]{.54\linewidth}
  \centering
  \includegraphics[width=9cm]{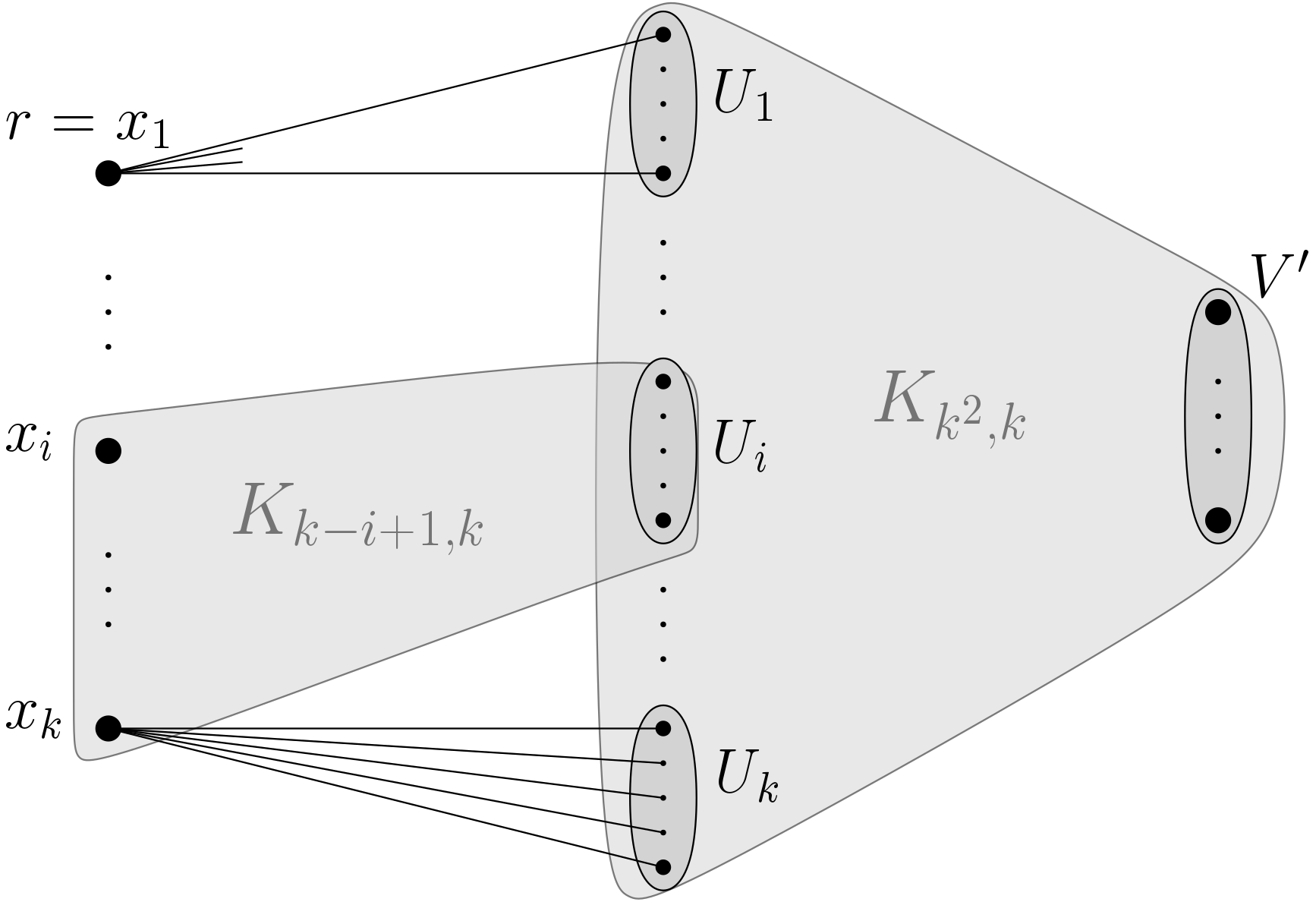}
  \caption{hard bipartite graph. $|U_1|=\dots=|U_n|=|V'|=k$. Any vertex is $U_i$ is adjacent to $V'\cup \{x_i,\dots,x_n\}$. 
\Greedy picks $V'\cup\{x_1,\dots,x_n\}$ while the maximum independent set is $\bigcup_i U_i$. The gap is $k/2=\Theta(\sqrt{n})$.
  \label{fig:hardbipartite}}
 \end{minipage}
\end{figure}

\begin{theo}
\begin{itemize}
\item For general graphs, \pr{MaxGreedy} is hard to approximate within a factor of $n^{1-\vep}$, for any constant $\vep > 0$, assuming $P\neq NP$.
\item For general graphs, \pr{MaxGreedy} is hard to approximate within a factor of $n/\log n$, assuming the exponential time hypothesis. 
\item For graphs with maximum degree $\Delta\ge 7$, \pr{MaxGreedy} is hard to approximate within a factor of $(\Delta+1)/3 - \calO(1/\Delta)- \calO(1/n)$, assuming $P\neq NP$.
\item For bipartite graphs, \pr{MaxGreedy} is hard to approximate within a factor of $n^{1/2-\vep}$, for any constant $\vep > 0$, assuming $P\neq NP$.
\end{itemize}
\end{theo}

Notice that $n$ refers here to the size of the graph considered. We give below additional details for each specific class of graphs.

\paragraph{General graphs.} Consider the hard graph $B$ given by Figure \ref{fig:hardgeneral}. For any $\vep >0$, choose its size $n$ such that $\vert \varphi \vert < n^\vep$. The size of $G$ is a polynomial in $n$. This implies that \pr{MaxGreedy} is hard to approximate to within $n^{1-\vep}$. Interestingly, Håstad proved that \pr{MIS} is also hard to approximate to within the same factor \cite{hastad1999}.

Further, making the \emph{exponential time} hypothesis  that there exists a constant $\epsilon>0$ such that $3$-SAT cannot be solved in time $2^{\epsilon n}$ implies that \pr{MaxGreedy} is hard to approximate to within a factor of $n/\log n$. To see this, one may use the same construction  with a bad graph of size $2^{\epsilon' |\varphi|}$ for a suitable value $\epsilon'>0$. In particular this suggests that a greedy algorithm will not do as well as Feige’s algorithm that achieves an $O(n/ (\log n)^3)$ approximation \cite{Feige04}.

\paragraph{Bounded degree graphs.} For graph with maximum degree $\Delta$, we use hard graphs described in the proof of Theorem \ref{theo:deltaapx}, see Figures \ref{fig:lowerboundanygreedy} and \ref{fig:delta1}. Recall that the structure of these graphs depends on the value $\Delta\mod 3$. As defined before, these hard graphs have several vertices of minimum degree, forming a clique $K$. Then, add a vertex $r''$, adjacent to all vertices in $K$ and a root vertex $r$ adjacent to $r''$. Now, these graphs have an unique root and the size of the greedy set and the maximum independent set only grew up by an additive constant factor. 
Since $\Delta\ge 7$, all vertices in $B$ have degree at least $5$ so that \Greedy will execute first all vertices in the anchor graph.

All vertices in $G$ have degree at most $\Delta$, except the root $r'$ that has degree equal to a least the number of clauses. To overcome this issue, we can replace edges $(r_C,r)$ by a tree of bounded degree, rooted in $r'$, such that the distance from and $r_C$ to $r$ is odd, and such that vertices at odd distance from $r'$ have smaller degree than the ones at even distance. We can finally add several cliques to increase the degree of vertices of this tree so that, they all have degree at least five. If at some point a vertex $r_C$ is picked by \Greedy --- meaning that the formula is not satisfied --- then all vertices at odd distance from $r'$ on the path from $r_C$ to $r$ will be picked in $S$, so that $r'\notin S$.

\paragraph{Bipartite graphs.} 
Figure \ref{fig:hardbipartite} presents an example of an hard bipartite graph, with a ratio $\Theta(\sqrt{n})$. When the root $r$ is removed from $B$, then the vertices in $U_1$ have minimum degree $2k-1$ while all vertices $x_1$ have degree at least $2k$. Therefore, in this situation the greedy set output is the union of all $U_i$, that is a maximum independent set in $B$.

To make the anchor graph bipartite, we first change the structure of gadgets $K_6$ used to increases the degree of some vertices. For instance, a simple bipartite complete graph $K_{6,6}$ fulfils exactly the same function. Then, one might still find some odd-cycles left due to edges between literal vertices and gadgets. These can be avoided using a NP-hard version of the satisfaction problem called \pr{Monotone} 3-SAT-4, where each variable appears four times and additionally that containing three literals, each clause contains only positive, or only negative literals. See \cite{monotone} for more details.

Since an optimal independent set can be computed in polynomial time in a bipartite graph, \Greedy is not a good algorithm for this class. However, this negative result suggests that even knowing a maximum independent set may not be helpful in order to design good greedy advises.

\section{Conclusions, future steps and open problems}

Our main technical contribution is a non-local payment scheme together with an inductive argument that can be embedded with greedy-style algorithms for MIS on bounded degree graphs. These techniques imply best possible approximation guarantees of greedy on subcubic graphs. We have also shown versatility of these techniques by proving (via simple proofs) that they imply close to best possible greedy guarantees on graphs with maximum degree $\Delta$ for any $\Delta$. Furthermore, they also imply improved fast approximation algorithms for the minimum vertex cover problem on bounded degree graphs. We have complemented these results by hardness results, showing that it is hard to compute good advice for the greedy MIS algorithms.

Our techniques have a potential to give further insights into the design of fast non-greedy algorithms that go past the greedy ``barrier'' in terms of approximation factors. Namely, a non-greedy algorithm can in certain situations choose degree-$3$ black vertex adjacent to the top white degree-$2$ vertex in the $5/4$-lower bound instances $H_i$ in Fig.~\ref{fig:lowerbound} with $H_0 = H_0'$. Observe that after such a non-greedy reduction, the algorithm can then follow greedy rules to finally compute an optimal solution on this instance. We have found a way of implementing such ``super-advice'' in $O(n^2)$ time. Moreover, our potential function and the inductive argument can be adapted to analyze the approximation guarantee of this method. This breaks through the $5/4$ lower bound of greedy and could even lead to approximations close to $6/5$. Recall, that $6/5$ is essentially the best currently known polynomial time approximation for subcubic MIS, achieved by local search algorithms which however have exorbitant running times. ``Super-advice'' could also be used to design non-greedy algorithms that go beyond the proved lower bound of $(\Delta + 1)/3$ for greedy on $\Delta$-bounded degree graphs.

Moreover, our techniques have a potential for further generalizations and applications. Our potential function and the inductive argument are quite general and they could be applied to other related problems on bounded degree graphs. Such general problem should have the following features: given a graph, the optimal solution should be ubiquitously ``distributed'' over the input graph, and therefore also a feasible solution should be computable sequentially/locally by ``choosing parts of the graph'', debt and loan should be definable on such a problem as problem specific, depending on the problem's constraints. Possible candidates are, for instance, the set packing and set covering problems with sets of bounded size and bounded element occurrences.

Finally, we also mention some more specific directions for further study. Can we obtain a greedy advise to design a $(\Delta+1)/3$-approximation for any value of $\Delta$? We think that it should be possible with our techniques, by a careful and refined analysis. For instance, we can already easily prove an $9/5$-approximation when $\Delta = 4$, and a $13/6$-approximation when $\Delta = 5$. This already improves on the previous ratios that follow from the known bound of $(\Delta+2)/3$ on any greedy: $2$ for $\Delta = 4$, and $7/3$ for $\Delta = 5$.

\bigskip

\noindent
{\bf Acknowledgment.} We would like to thank Magn{\'{u}}s Halld{\'o}rsson for his help with verifying our example and for his explanations to the results about greedy algorithms.

\bibliographystyle{siam}
\bibliography{main}

\begin{thebibliography}{10}

\bibitem{AlimontiK00}
{\sc P.~Alimonti and V.~Kann}, {\em Some {APX}-completeness results for cubic
  graphs}, Theor. Comput. Sci., 237 (2000), pp.~123--134.

\bibitem{AKS2011}
{\sc P.~Austrin, S.~Khot, and M.~Safra}, {\em Inapproximability of vertex cover
  and independent set in bounded degree graphs}, Theory of Computing, 7 (2011),
  pp.~27--43.

\bibitem{Bansal15}
{\sc N.~Bansal}, {\em Approximating independent sets in sparse graphs}, in
  Proceedings of the Twenty-Sixth Annual {ACM-SIAM} Symposium on Discrete
  Algorithms, {SODA} 2015, San Diego, CA, USA, January 4-6, 2015, 2015,
  pp.~1--8.

\bibitem{Bansal:2015:LTF:2746539.2746607}
{\sc N.~Bansal, A.~Gupta, and G.~Guruganesh}, {\em On the lov\'{a}sz theta
  function for independent sets in sparse graphs}, in Proceedings of the
  Forty-seventh Annual ACM Symposium on Theory of Computing, STOC '15, New
  York, NY, USA, 2015, ACM, pp.~193--200.

\bibitem{Berge1973}
{\sc C.~Berge}, {\em Graphs and hypergraphs}, North-Holland, Ansterdam, 1973.

\bibitem{Berman1999}
{\sc P.~Berman and T.~Fujito}, {\em On approximation properties of the
  independent set problem for low degree graphs}, Theory of Computing Systems,
  32 (1999), pp.~115--132.

\bibitem{Berman:1994:AMI:314464.314570}
{\sc P.~Berman and M.~F\"{u}rer}, {\em Approximating maximum independent set in
  bounded degree graphs}, in Proceedings of the Fifth Annual ACM-SIAM Symposium
  on Discrete Algorithms, SODA '94, Philadelphia, PA, USA, 1994, Society for
  Industrial and Applied Mathematics, pp.~365--371.

\bibitem{BODLAENDER1997101}
{\sc H.~L. Bodlaender, D.~M. Thilikos, and K.~Yamazaki}, {\em It is hard to
  know when greedy is good for finding independent sets}, Information
  Processing Letters, 61 (1997), pp.~101--106.

\bibitem{Boppana1992}
{\sc R.~Boppana and M.~M. Halld{\'o}rsson}, {\em Approximating maximum
  independent sets by excluding subgraphs}, BIT Numerical Mathematics, 32
  (1992), pp.~180--196.

\bibitem{CARO1996137}
{\sc Y.~Caro, A.~Sebő, and M.~Tarsi}, {\em Recognizing greedy structures},
  Journal of Algorithms, 20 (1996), pp.~137 -- 156.

\bibitem{SOChan16}
{\sc S.~O. Chan}, {\em Approximation resistance from pairwise-independent
  subgroups}, J. {ACM}, 63 (2016), pp.~27:1--27:32.

\bibitem{Chlebiks2003}
{\sc M.~Chlebik and J.~Chlebikova}, {\em Inapproximability results for bounded
  variants of optimization problems}, 12 2003, pp.~27--38.

\bibitem{Chlebik2004}
{\sc M.~Chleb{\'i}k and J.~Chleb{\'i}kov{\'a}}, {\em On approximability of the
  independent set problem for low degree graphs}, in Structural Information and
  Communication Complexity, R.~Kr{\'a}lovic and O.~S{\'y}kora, eds., Berlin,
  Heidelberg, 2004, Springer Berlin Heidelberg, pp.~47--56.

\bibitem{10.1007/978-3-540-27796-5_5}
\leavevmode\vrule height 2pt depth -1.6pt width 23pt, {\em On approximability
  of the independent set problem for low degree graphs}, in Structural
  Information and Communication Complexity, R.~Kr{\'a}lovic and O.~S{\'y}kora,
  eds., Berlin, Heidelberg, 2004, Springer Berlin Heidelberg, pp.~47--56.

\bibitem{monotone}
{\sc A.~Darmann and J.~D{\"{o}}cker}, {\em Monotone 3-sat-4 is np-complete},
  CoRR, abs/1603.07881 (2016).

\bibitem{DEMANGE1997105}
{\sc M.~Demange and V.~Paschos}, {\em Improved approximations for maximum
  independent set via approximation chains}, Applied Mathematics Letters, 10
  (1997), pp.~105 -- 110.

\bibitem{Erdos1970}
{\sc P.~Erd\H{o}s}, {\em On the graph theorem of {T}ur{\'a}n}, Mat. Lapok, 21
  (1970), pp.~249--251.

\bibitem{Feige04}
{\sc U.~Feige}, {\em Approximating maximum clique by removing subgraphs},
  {SIAM} J. Discrete Math., 18 (2004), pp.~219--225.

\bibitem{FriezeS94}
{\sc A.~M. Frieze and S.~Suen}, {\em On the independence number of random cubic
  graphs}, Random Struct. Algorithms, 5 (1994), pp.~649--664.

\bibitem{GAREY1976237}
{\sc M.~Garey, D.~Johnson, and L.~Stockmeyer}, {\em Some simplified
  {NP}-complete graph problems}, Theoretical Computer Science, 1 (1976),
  pp.~237 -- 267.

\bibitem{garey}
{\sc M.~R. Garey, D.~S. Johnson, and L.~Stockmeyer}, {\em Some simplified
  np-complete problems}, in Proceedings of the Sixth Annual ACM Symposium on
  Theory of Computing, STOC '74, New York, NY, USA, 1974, ACM, pp.~47--63.

\bibitem{HajnalS1970}
{\sc A.~Hajnal and E.~Szemer{\'e}di}, {\em Proof of a conjecture of {P.}
  {E}rd{\H{o}}s}, in Combinatorial Theory and its Applications, P.~Erd\H{o}s,
  A.~R{\'e}nyi, and V.~S{\'o}s, eds., North-Holland, Amsterdam, 1970,
  pp.~601--623.

\bibitem{Halldorsson2019}
{\sc M.~M. Halld{\'o}rsson}, {\em Private communication. Erratum: {\tt
  https://www.ru.is/$\sim$mmh/papers/hy95-errata.html}}, March, 2019.

\bibitem{HalldorssonR1994}
{\sc M.~M. Halld{\'{o}}rsson and J.~Radhakrishnan}, {\em Greed is good:
  approximating independent sets in sparse and bounded-degree graphs}, in
  Proceedings of the Twenty-Sixth Annual {ACM} Symposium on Theory of
  Computing, {STOC} 1994, Montr{\'{e}}al, Qu{\'{e}}bec, Canada, 23-25 May 1994,
  1994, pp.~439--448.

\bibitem{HalldorssonRadha1994}
\leavevmode\vrule height 2pt depth -1.6pt width 23pt, {\em Improved
  approximations of independent sets in bounded-degree graphs}, in Algorithm
  Theory - {SWAT} '94, 4th Scandinavian Workshop on Algorithm Theory, Aarhus,
  Denmark, July 6-8, 1994, Proceedings, 1994, pp.~195--206.

\bibitem{Halldorsson1994Removal}
\leavevmode\vrule height 2pt depth -1.6pt width 23pt, {\em Improved
  approximations of independent sets in bounded-degree graphs via subgraph
  removal.}, Nord. J. Comput., 1 (1994), pp.~475--492.

\bibitem{Halldorsson1997}
{\sc M.~M. Halld{\'o}rsson and J.~Radhakrishnan}, {\em Greed is good:
  Approximating independent sets in sparse and bounded-degree graphs},
  Algorithmica, 18 (1997), pp.~145--163.

\bibitem{10.1007/BFb0015418}
{\sc M.~M. Halld{\'o}rsson and K.~Yoshihara}, {\em Greedy approximations of
  independent sets in low degree graphs}, in Algorithms and Computations,
  J.~Staples, P.~Eades, N.~Katoh, and A.~Moffat, eds., Berlin, Heidelberg,
  1995, Springer Berlin Heidelberg. Full version in: {\tt
  http://www.ru.is/$\sim$mmh/raunvis/Papers/yoshi.pdf}, pp.~152--161.

\bibitem{Halperin02}
{\sc E.~Halperin}, {\em Improved approximation algorithms for the vertex cover
  problem in graphs and hypergraphs}, {SIAM} J. Comput., 31 (2002),
  pp.~1608--1623.

\bibitem{HOCHBAUM1983243}
{\sc D.~S. Hochbaum}, {\em Efficient bounds for the stable set, vertex cover
  and set packing problems}, Discrete Applied Mathematics, 6 (1983), pp.~243 --
  254.

\bibitem{hastad1999}
{\sc J.~Håstad}, {\em Clique is hard to approximate within $n^{1 -\epsilon}$},
  Acta Math., 182 (1999), pp.~105--142.

\bibitem{DSJohnson1974}
{\sc D.~S. Johnson}, {\em Worst case behavior of graph coloring algorithms}, in
  Proceedings of the 5th Southeastern Conference on Combinatorics, Graph Theory
  and Computing, Congressus Numerantium X, 1974, pp.~513–--527.

\bibitem{Karp72}
{\sc R.~M. Karp}, {\em Reducibility among combinatorial problems}, in
  Proceedings of a symposium on the Complexity of Computer Computations, held
  March 20-22, 1972, at the {IBM} Thomas J. Watson Research Center, Yorktown
  Heights, New York, {USA}, 1972, pp.~85--103.

\bibitem{Khanna1998}
{\sc S.~Khanna, R.~Motwani, M.~Sudan, and U.~Vazirani}, {\em On syntactic
  versus computational views of approximability}, SIAM Journal on Computing, 28
  (1998), pp.~164--191.

\bibitem{LOVASZ1975269}
{\sc L.~Lov{\'a}sz}, {\em Three short proofs in graph theory}, Journal of
  Combinatorial Theory, Series B, 19 (1975), pp.~269 -- 271.

\bibitem{McDiarmid84}
{\sc C.~McDiarmid}, {\em Colouring random graphs}, Annals of Operations
  Research, 1 (1984), pp.~183--200.

\bibitem{Nemhauser1975}
{\sc G.~L. Nemhauser and L.~E. Trotter}, {\em Vertex packings: Structural
  properties and algorithms}, Mathematical Programming, 8 (1975), pp.~232--248.

\bibitem{Plummer70}
{\sc M.~D. Plummer}, {\em Some covering concepts in graphs}, Journal of
  Combinatorial Theory, 8(1) (1970), pp.~91--98.

\bibitem{Plummer93}
\leavevmode\vrule height 2pt depth -1.6pt width 23pt, {\em Well-covered graphs:
  a survey}, Quaestiones Mathematicae, 16(3) (1993), pp.~253--287.

\bibitem{Shearer83}
{\sc J.~B. Shearer}, {\em A note on the independence number of triangle-free
  graphs}, Discrete Mathematics, 46 (1983), pp.~83 -- 87.

\bibitem{Simon90}
{\sc H.~U. Simon}, {\em On approximate solutions for combinatorial optimization
  problems}, {SIAM} J. Discrete Math., 3 (1990), pp.~294--310.

\bibitem{Turan1941}
{\sc P.~Tur{\'a}n}, {\em An extremal problem in graph theory}, Mat. Fiz. Lapok,
  48 (1941), pp.~436--452.

\bibitem{Wei1981}
{\sc V.~Wei}, {\em A lower bound on the stability number of a simple graph},
  Bell Laboratories Technical Memorandum, 81-11217-9 (Murray Hill, NJ, 1981).

\bibitem{Zuckerman:2006:LDE:1132516.1132612}
{\sc D.~Zuckerman}, {\em Linear degree extractors and the inapproximability of
  max clique and chromatic number}, in Proceedings of the Thirty-eighth Annual
  ACM Symposium on Theory of Computing, STOC '06, New York, NY, USA, 2006, ACM,
  pp.~681--690.

\end{thebibliography}

\newpage


\end{document}